\documentclass[accepted]{uai2021} 

\usepackage{natbib} 
\usepackage{mathtools} 
\usepackage{booktabs} 
\usepackage{tikz} 

\usepackage{fullpage}
\usepackage[utf8]{inputenc}
\usepackage{graphicx}
\usepackage{amsmath}
\usepackage{amsfonts}
\usepackage{mathtools}
\usepackage[mathscr]{euscript}
\usepackage[caption=false]{subfig}
\usepackage{algorithm}
\usepackage{algpseudocode}
\usepackage{amsthm}
\usepackage{amssymb}
\usepackage{rotating}
\usepackage{xcolor}
\usepackage{enumerate}
\usepackage{multicol}
\usepackage{multirow}
\usepackage[normalem]{ulem}

\usepackage{xr}

\makeatletter
\newcommand*{\addFileDependency}[1]{
	\typeout{(#1)}
	\@addtofilelist{#1}
	\IfFileExists{#1}{}{\typeout{No file #1.}}
}
\makeatother

\newcommand*{\myexternaldocument}[1]{%
	\externaldocument{#1}%
	\addFileDependency{#1.tex}%
	\addFileDependency{#1.aux}%
}

\myexternaldocument{main_supp}

\def\I{{\mathbb I}}
\def\P{{\mathbb P}}
\def\pr{\mathbb{P}}

\def\X{{\mathbf{X}}}
\def\x{{\mathbf{x}}}
\def\z{{\mathbf{z}}}
\def\y{{\mathbf{y}}}
\def\Y{{\mathbf{Y}}}

\def\z{{\mathbf{z}}} 
\def\E{{\mathbb E}}

\def\PIT{{\text{PIT}}}
\def\HPD{{\text{HPD}}}

\newcommand{\gvn}{|}

\newtheorem{thm}{Theorem}

\newtheorem{Corollary}{Corollary}
\newtheorem{Definition}{Definition}
\newtheorem{Remark}{Remark}

\newcommand{\codecomment}[1]{\textbf{\color{black}// #1}}

\newif\ifdraft
\drafttrue 

\newcommand{\annlee}[1]{\ifdraft \textbf{\color{magenta}[Ann: #1]} \else \fi} 

\title{Diagnostics for Conditional Density Models and Bayesian Inference Algorithms}

\author[1]{\href{mailto:David Zhao <davidzhao@stat.cmu.edu>?Subject=Your UAI 2021 paper}{David~Zhao}{}} 
\author[1]{Niccol\`o~Dalmasso}
\author[2]{Rafael~Izbicki}
\author[1]{Ann~B.~Lee}
\affil[1]{%
	Department of Statistics \& Data Science\\
	Carnegie Mellon University\\
	Pittsburgh, Pennsylvania, USA
}
\affil[2]{%
	Department of Statistics\\
	Federal University of S\~{a}o Carlos (UFSCar)\\
	S\~{a}o Carlos, Brazil
}

\begin{document}
\maketitle


\begin{abstract} 
	There has been growing interest in the AI community for precise uncertainty quantification. Conditional density models $f(y|\x)$, where $\x$ represents  potentially high-dimensional features, are an integral part of uncertainty quantification in prediction and Bayesian inference. However,  it is challenging to assess conditional density estimates and gain insight into modes of failure. While existing  diagnostic tools can determine whether an approximated conditional density  is compatible overall with a data sample, they lack  a principled framework for identifying, locating, and interpreting the nature of   statistically significant discrepancies over the entire feature space. In this paper, we present  rigorous and easy-to-interpret diagnostics such as (i) the ``Local Coverage Test'' (LCT), which   distinguishes an arbitrarily misspecified model from the true conditional density of the sample, and (ii) ``Amortized Local P-P plots'' (ALP) which can quickly provide interpretable graphical summaries of distributional differences at any location $\x$ in the feature space. Our validation procedures scale to high dimensions and can potentially adapt to any type of data at hand. We demonstrate the effectiveness of LCT and ALP through a simulated experiment and  applications to prediction and parameter inference for image data. 
	
\end{abstract}

\section{Introduction}\label{sec:intro}
There has been growing interest in the AI community for precise uncertainty quantification (UQ), with conditional density models playing a key role in UQ in prediction and Bayesian inference. For instance, the conditional density  $f(y|\x)$ of the response variable $y$ given features $\x$ can be used to build predictive regions for $y$, which are more informative than point predictions. Indeed, in prediction settings, $f$ provides a full account of the uncertainty in the outcome $y$ given new observations $\x$. Conditional densities are also central to Bayesian parameter inference, where the posterior distribution  $f(\theta|\x)$ is key to quantifying
uncertainty about the parameters $\theta$ of interest after observing data $\x$. 

Recently, a large body of work in machine learning has been developed for estimating conditional densities $f$ for all possible values of $\x$, or to generate predictions that follow the unknown conditional density (see \citealt{uria2014density,sohn2015generativecde, papamakarios2017maf, deisenroth2016gpcde, papamakarios2019normalizing} and references therein). With the advent of high-precision data and simulations, simulation-based inference (SBI; \cite{cranmer2020sbi}) has also played a growing role in disciplines ranging from  physics, chemistry and engineering to the biological and social sciences.  The SBI category includes 
machine-learning based methods to learn an explicit surrogate model of the posterior \citep{marin2016abcrf, papamakarios2016snpea, lueckmann2017snpeb, chen2019gaussiancopula, izbicki2019abc, greenberg2019posterior}. 


Inevitably, any downstream analysis in predictive modeling or Bayesian inference depends on 
the trustworthiness of the assumed conditional density model. 
Validating such models can 
be challenging, especially for high-dimensional or mixed-type data $\x$. There does not currently exist a comprehensive and rigorous 
set of diagnostics that describe, for all 
values of $\x$, the quality of fit of a conditional density model.

\textbf{Related work.} 
Large AI models, such as deep generative autoregressive models or Bayesian networks, are typically fit using global loss functions like the Kullback-Leibler divergence or the $L^2$ loss  \citep{izbicki2017photo, rothfuss2019cdepractices}. Loss functions are useful for training models but only provide relative comparisons of overall model fit. Hence, a practitioner may not  know whether he or she should keep looking for better models  (using larger training samples, training times, etc.), or if the current estimate is ``close enough''. 
Another line of work assesses goodness-of-fit of a conditional density model via a two-sample test that compare samples from $\widehat f$ and $f$. Earlier tests involve a conditional version of the standard Kolmogorov test \citep{andrews1997condtest, zheng2000condtest} in one dimension, or are tailored to specific families of conditional densities \citep{stute2002condtest, moreira2003condtest}.  Recently, \cite{jitkrittum2020cde} developed a fast kernel-based approach that can also identify local regions of poor fit. While these tests are consistent,
they do not provide insight on how the distributions of $\widehat f$ and $f$ differ locally. Kernel approaches also require the user to 
specify an appropriate kernel  and tuning parameters, which can be challenging in practice. Finally, existing diagnostics that do describe the nature of inconsistencies between $\widehat f$ and $f$ only test for a form of overall coherence between a data-averaged conditional (posterior) distribution and its marginal (prior) distribution. Typically, they compute probability integral transform (PIT) values \citep{cook2006validating,freeman2017photoz,talts2018validating,disanto2018cmdn}. While informative, these diagnostics were originally developed for assessing \textit{unconditional} density models \citep{gan1990pit}. As such, they are known to fail to detect some clearly misspecified conditional models including models that ignore the dependence on the covariates altogether  \citep{schmidt2020photoz}. 
(Our Theorem \ref{thm:uniform} 
details different failure modes of existing diagnostics.)

\textbf{Contribution and novelty.} Our work provides diagnostic tools for UQ and calibration of predictive models that provide insight in simple, explainable terms like coverage, bias, dispersion, and multimodality in $y$ (output of interest) as a function of $\x$ (observed inputs). Having interpretable diagnostics is crucial for scientific collaborators and end users to build trust in large AI models.

Existing diagnostics for conditional density models cannot detect every kind of misspecified model and give insight into local quality of fit at any given $\x$. Our method quantifies deviations between actual and nominal coverage in $y$. It (i) detects arbitrarily misspecified models and (ii) assesses and visualizes quality of fit anywhere in feature space, even at points without observed data, in terms of easy-to-explain diagnostics. To the best of our knowledge, no other method in the literature provides both consistency and diagnostics for complex high-dimensional data. 

To enrich our vocabulary for desired properties of CDEs, we begin our paper by defining global and local consistency (see Definitions \ref{def:global_cons} and \ref{def:local_cons}, respectively). We then describe our diagnostic framework, which 
has three main components:
\begin{itemize}
	\vspace{-1mm}
	\item \textbf{[GCT - Global Coverage Test]} A statistical hypothesis test that can distinguish \emph{any} misspecified density model from the true conditional density.  (This is a test of global consistency.) 
	\item \textbf{[LCT - Local Coverage Test]}  A statistical hypothesis test that identifies \emph{where} in the feature space the model fits poorly. (This is a test of local consistency.) 
	\item \textbf{[ALP - Amortized  Local  P-P  plots]} Interpretable graphical summaries of the fitted model that show \emph{how} it deviates from the true
	density  at  any  location  in  feature space (see Figure \ref{fig:PPplot_interpret} for examples). 
	We also provide amortized PIT histograms that contain the same information as ALPs but in a different format (see Appendix C for details).
\end{itemize}
Our diagnostics are easy and fast to compute, and can identify, locate, and interpret the nature of  (statistically significant) discrepancies over the entire feature space.  At the heart of our approach is the realization that the local coverage of a CDE model is itself a conditional probability (see Equation \ref{eq:r_alpha}) that often varies smoothly with $\x$. Hence, we can estimate the local coverage at any given $\x$ by leveraging a suitable regression method using sample points in a neighborhood of $\x$. Thanks to the impressive arsenal of existing regression methods, we can adapt to different types of potentially high-dimensional data to obtain computationally and statistically efficient validation. Finally, because we specifically evaluate local coverage (rather than other types of discrepancies), the practitioner can ``zoom in'' on statistically significant local discrepancies flagged by the LCT, and identify common modes of failure in the fitted conditional density (see Figures \ref{fig:toy_example1_badfit}-\ref{fig:galsim_example} for examples). 

All code used to produce our experiments is available at \url{https://github.com/zhao-david/CDE-diagnostics}. We have also included an installable Python package \texttt{cde-diagnostics} with a detailed tutorial.

\section{Existing Diagnostics are Insensitive to Covariate Transformations}
\label{sec:existing_diag_insens}

{\bf Notation.} Let $\mathcal{D}=\{(\X_1,Y_1),\ldots,(\X_n,Y_n)\}$ denote an i.i.d. sample from  $F_{\X,Y}$, the joint distribution of $(\X,Y)$ for a random variable $Y \in \mathcal{Y} \subseteq \mathbb{R}$  (in Section \ref{sec:multivariate}, $Y$ is multivariate), and a random vector $\X \in \mathcal{X} \subseteq \mathbb{R}^d$. In a prediction setting, $\mathcal{D}$ represents a hold-out set not used to train $\widehat f$. In a Bayesian setting, $Y$ represents the parameter of interest (sometimes also denoted with $\theta$), and each element of $\mathcal{D}$ is obtained by first drawing $Y_i$ from the prior distribution, and then drawing $\X_i$ from the statistical model of $\X|Y_i$.

Ideally, a test should be able to distinguish \textit{any} given 
alternative conditional density model $\widehat{f}(y|\x)$ from the true density $f(y|\x)$, as well as locate discrepancies in the feature space $\mathcal{X}$. More precisely, a test should be able to identify what we in this section define as global and local consistency.
\begin{Definition}[\textbf{Global Consistency}]
	\label{def:global_cons}
	An estimate $\widehat f(y|\x)$ is globally consistent with the density 
	$ f(y|\x)$
	if the following null hypothesis holds:
	\begin{align}
	\label{eq:global_null}
	H_0:\widehat f(y|\x)= f(y|\x) \mbox{ for every } \x \in \mathcal{X} \mbox{ and } y \in \mathcal{Y}.
	\end{align}
\end{Definition}
\vspace{-0.25cm}
Note that $\widehat f$ is a particular fixed conditional density estimate, and we test whether samples from $\widehat f$ are consistent with samples from $f$.
Existing diagnostics typically validate density models by computing PIT values on independent data, which were not used to estimate $\widehat f(y|\x)$:

\begin{Definition}[\textbf{PIT}] Fix $\x \in \mathcal{X}$ and  $y \in \mathcal{Y}$.
	The probability integral transform of  $y$ at $\x$, as modeled by the conditional density estimate $\widehat{f}(y|\x)$, is
	\begin{equation}
	\label{eq:pit}
	\PIT(y; \x) = \int_{-\infty}^{y} \widehat{f}(y' \gvn \x) dy'. 
	\end{equation}
\end{Definition}
\vspace{-0.25cm}
See Figure \ref{fig:pit_hpd_visualization}, top panel for an illustration of this calculation.

\begin{Remark}
	For implicit models of $\widehat f(y|\x)$
	(that is, generative models that via e.g. MCMC can sample from, but not directly evaluate  $\widehat f$), we can approximate the PIT values by forward-simulating data: For fixed $\x \in \mathcal{X}$ and  $y \in \mathcal{Y}$, draw $Y_1,\ldots,Y_L \sim \widehat{f}(\cdot|\x)$. Then, approximate $\PIT(y; \x)$ via  the cumulative sum $L^{-1} \sum_{i=1}^L  \I(y_i\leq y)$.
\end{Remark}

{\em If} the conditional density model $\widehat f(y|\x)$ is globally consistent, then the  \PIT\ values are uniformly distributed. More precisely, if $H_0$ (Equation \ref{eq:global_null}) is true, then  the random variables
$\PIT(Y_1; \X_1),\ldots,\PIT(Y_n; \X_n) \stackrel{i.i.d.}{\sim} \textrm{Unif}(0,1)$. This result is often used to test goodness-of-fit of conditional density models in practice \citep{cook2006validating, bordoloi2010photoz, tanaka2018photoz}. 

Our first point is that unfortunately, such random variables can be uniformly distributed even if global consistency does not hold. This is shown in the following theorem. 
\begin{thm}[\textbf{Insensitivity to Covariate Transformations}]
	\label{thm:uniform}
	Suppose there exists a function $g:\mathcal{X} \longrightarrow \mathcal{Z}$, where $\mathcal{Z} \subseteq \mathbb{R}^k$ for some $k$, that satisfies
	\begin{equation}
	\widehat f(y|\x) = f(y|g(\x)).  \label{eq:reduction}
	\end{equation}
	Let $(\X,Y) \sim F_{\X,Y}$. Then $\PIT(Y;\X) \sim \textrm{Unif}(0,1)$.
\end{thm}
\vspace{-0.25cm}
Many models naturally lead to estimates that  could satisfy the condition in Equation \ref{eq:reduction}, even  without being  globally consistent.  In fact, clearly misspecified models $\widehat f$ can yield uniform PIT values and ``pass'' an associated goodness-of-fit test regardless of the sample size.  For example: if $\widehat f(y|\x)$ is based on  a linear model, then  $\widehat f(y|\x)$ will by construction depend on $\x \in \mathbb{R}^d$ only through $g(\x):=\beta^T \x$ for some $\beta \in \mathbb{R}^d$. As a result, we could have $\widehat f(y|\x) = f(y|g(\x))$ even when $\widehat f(y|\x)$ is potentially very different from  $f(y|\x)$. 
As another example, a conditional density estimator that performs variable selection \citep{shiga2015cdeauto, izbicki2017converting, dalmasso2020conditional} 
could satisfy  $\widehat f(y|\x) = f(y|g(\x))$ for $g(\x):= (\x)_S$, where $S \subset \{1,\ldots,d\}$ is a subset of the covariates.  A test of the overall uniformity of \PIT \ values is no guarantee that we are correctly modeling the relationship between $y$ and the predictors $\x$; see Figure \ref{fig:PIT_fail} for an illustration.

Our second point is that current diagnostics also do not pinpoint the locations in feature space $\mathcal{X}$ where the estimates of $f$ should be improved. Hence, in addition to global consistency, we need diagnostics that test the following property:
\begin{Definition}[\textbf{Local Consistency}]
	\label{def:local_cons}
	Fix $\x \in \mathcal{X}$. An estimate $\widehat f(y|\x)$ is locally consistent  with the density 
	$ f(y|\x)$ at fixed $\x$
	if the following null hypothesis holds:
	\begin{align}
	\label{eq:local_null}
	H_0(\x):\widehat f(y|\x)= f(y|\x) \mbox{ for every } y \in \mathcal{Y}.
	\end{align}
\end{Definition}
\vspace{-0.25cm}
In the next section, we introduce new diagnostics that are able to test whether a conditional density model $\widehat f$ is both globally and locally consistent with the underlying conditional distribution $f$ of the data. Our diagnostics are still based on PIT, and hence retain the properties (e.g., interpretability, ability to provide graphical summaries, and so on) that have made PIT a popular choice in model validation.

\section{New Diagnostics Test Local and Global Consistency}
\label{sec:new_diagnostics}

Our new diagnostics rely on the following key result:
\begin{thm}[\textbf{Local Consistency and Pointwise Uniformity}]
	\label{thm:local_uniform}
	For any $\x \in \mathcal{X}$, the local null hypothesis $H_0(\x):\widehat{f}(\cdot|\x)=f(\cdot|\x)$ holds if, and only if, the distribution of $\PIT(Y;\x)$ given $\x$ is uniform over $(0,1)$. 
\end{thm}

Theorem \ref{thm:local_uniform} implies that if we had a sample of $Y$'s at the fixed location $\x$, we could test the local consistency (Definition \ref{def:local_cons}) of $\widehat f$ by determining whether the sample's PIT values come from a uniform distribution. In addition, for global consistency we need local consistency at every $\x \in \mathcal{X}$. Clearly, such a testing procedure would not be practical: typically, we have data of the form  $(\X_1,Y_1)
,\ldots,( \X_n,Y_n)$ with at most one observation at any given $\x \in \mathcal{X}$.

Our solution is to instead address this problem as a regression: for fixed $\alpha \in (0,1)$, we consider the cumulative distribution function (CDF) of PIT at $\x$,
\begin{equation}
\label{eq:r_alpha}
r_{\alpha}(\x) := \pr\left(\PIT(Y;\x) < \alpha|\x \right),
\end{equation}
which is the regression of the random variable $W^{\alpha}:=\I(\PIT(Y;\X)< \alpha)$ on $\X$. 

From Theorem \ref{thm:local_uniform}, it follows that the estimated density is locally consistent at $\x$ if and only if $r_{\alpha}(\x) = \alpha$ for every $\alpha$:
\begin{Corollary}
	\label{cor:diagnostics}
	Fix $\x \in \mathcal{X}$. Then $r_{\alpha}(\x)=\alpha$ for every $\alpha \in (0,1)$ if, and only if, $\widehat{f}(y|\x)=f(y|\x)$ for every $y \in \mathcal{Y}$.
\end{Corollary}

Our new diagnostics are able to test for both local and global consistency. They rely on the simple idea of estimating $r_{\alpha}(\x)$  and then evaluating how much it deviates from $\alpha$ (see Section \ref{sec:our_tests}). Note that $$\PIT(Y;\x)<\alpha \iff Y \in (-\infty,\widehat q_\alpha(\x))$$ where $\widehat q_\alpha(\x)$ is the $\alpha$-quantile of $\widehat f$. That is, $r_\alpha(\x)$ 
assesses the local
level-$\alpha$ \textbf{coverage} of $\widehat f$ at $\x$.  In Section \ref{sec:pp}, we explore the connection between test statistics and coverage, for interpretable descriptions of how conditional density models $\widehat f$ may fail to approximate the true conditional density $f$.

\subsection{Local and Global Coverage Tests}
\label{sec:our_tests}
Our procedure for testing local and global consistency is very simple and can be adapted to different types of data. For an i.i.d. test sample $(\X_1,Y_1),\ldots,(\X_n,Y_n)$ from $F_{\X,Y}$ (which was not used to construct $\widehat f$), we compute $W^{\alpha}_i:=\I( \PIT(Y_i;\X_i) < \alpha)$. To estimate the coverage $r_{\alpha}(\x)$ (Equation \ref{eq:r_alpha}) for any $\x \in \mathcal{X}$, we then simply regress $W$ on $\X$ using the transformed data   $(\X_1,W_1),\ldots,(\X_n, W_n)$. Numerous classes of regression estimators can be used, from kernel smoothers to random forests to neural networks.

To test local consistency  (Definition \ref{def:local_cons}), we introduce the {\em Local Coverage Test} (LCT) with the test statistic

\begin{equation*}
T(\x):=\frac{1}{|G|}\sum_{\alpha \in G} (\widehat r_{\alpha}(\x)-\alpha)^2,
\end{equation*} 
\vspace{-0.25cm}

where $\widehat r_{\alpha}$ denotes the regression estimator and $G$ is a grid of $\alpha$ values.
Large values of $T(\x)$ indicate a large discrepancy  between $\widehat f$ and $f$ at $\x$
in terms of coverage, and Corollary \ref{cor:diagnostics} links coverage to consistency. To decide on the correct cutoff for rejecting $H_0(\x)$, we use a Monte Carlo technique that simulates $T(\x)$ under $H_0$. Algorithm \ref{alg:p_values} details our procedure. For the LCT, note that we are performing multiple hypothesis tests at different locations $\x$. After obtaining LCT p-values, we advocate using a method like Benjamini-Hochberg to control the false discovery rate.

Similarly, we can also test global consistency (Definition \ref{def:global_cons}) with a Monte Carlo strategy. Algorithm \ref{alg:p_values_gct} in Supp. Mat. B. details our procedure. We introduce the {\em Global Coverage Test} (GCT) based on the following test statistic:

\vspace{-0.25cm}
\begin{equation*}
S:= \frac{1}{n} \sum_{i=1}^n T(\X_i).
\end{equation*}
\vspace{-0.25cm}

We recommend performing the global test first and, if the global null is rejected, investigating further with local tests.
Empirically, we have found that the power of our tests is related to the MSE (a measurable quantity) of the regression method we use. This observation is in line with similar results in \citet[Theorems 3.3 and 4.1]{kim2019globallocal}.
Hence, as a practical strategy, we maximize power by choosing the regression model with the smallest MSE on validation data.

\setlength{\dbltextfloatsep}{14pt}
\begin{algorithm}[!ht]
	\caption{
		P-values for Local Coverage Test
	}\label{alg:p_values}
	\algorithmicrequire \ {\small  conditional density model $\widehat f$; test data $\{\X_i,Y_i\}_{i=1}^n$; test point $\x \in \mathcal{X}$; regression estimator $\widehat r$; 
		grid $G$ of of $\alpha$ values in $(0,1)$; number of null training samples $B$}\\
	\algorithmicensure \ {\small estimated p-value $\widehat{p}(\x)$ for any $\x \in \mathcal{X}$}
	\begin{algorithmic}[1]
		\State \codecomment{Compute test statistic at $\x$:}
		\State Compute values $\PIT(Y_1;\X_1),\ldots,\PIT(Y_n;\X_n)$
		\For{$\alpha$ in $G$}
		\State Compute indicators $W^{\alpha}_1,\ldots,W^{\alpha}_n$
		\State Train regression method $\widehat r_{\alpha}$ on $\{\X_i, W^{\alpha}_i\}_{i=1}^n$
		\EndFor
		\State Compute test statistic $T(\x)$
		\State \codecomment{Recompute test statistic under null distribution:}
		\For{$b$ in $1,\ldots,B$}
		\State Draw $U_1^{(b)},\ldots,U_n^{(b)} \sim \textrm{Unif}[0,1]$.
		\For{$\alpha$ in $G$}
		\State Compute indicators $ \{ W^{(b)}_{\alpha,i} = \I( U_i^{(b)} < \alpha)\}_{i=1}^n $
		\State Train regression method $\widehat r_{\alpha}^{(b)}$ on $\{\X_i, W^{(b)}_{\alpha,i}\}_{i=1}^n$
		\EndFor
		\State Compute $\displaystyle T^{(b)}(\x):=\frac{1}{|G|}\sum_{\alpha \in G} (\widehat r_{\alpha}^{(b)}(\x)-\alpha)^2$
		\EndFor
		\\
		\textbf{return}  $\displaystyle \widehat{p}(\x):=\frac{1}{B}\sum_{b=1}^B\I\left(T(\x) < T^{(b)}(\x)\right)$
	\end{algorithmic}
\end{algorithm}

\subsection{Amortized local P-P plots}
\label{sec:pp}
Our diagnostic framework does not just give us the ability to identify deviations from local consistency in different parts of the feature space $\mathcal{X}$. It also provides us with insight into the nature of such deviations at any given location $\x$.
For unconditional density models, data scientists have long favored using P-P plots (which plot two cumulative distribution functions against each other) to assess how closely a density model agrees with actual observed data. What makes our work unique is that we are able to construct ``amortized local P-P plots'' (ALPs) with similar interpretations to assess {\em conditional} density models over the entire feature space.

Figure \ref{fig:PPplot_interpret} illustrates how a local P-P plot of $\widehat r_{\alpha}(\x)$ against $\alpha$ (that is, the estimated CDF against the true CDF at $\x$) can identify different types of deviations in a conditional density model. For example, positive or negative bias in the estimated density $\widehat f$ relative to $f$ leads to P-P plot values that are too high or too low, respectively. We can also easily identify overdispersion or underdispersion of $\widehat f$ from an ``S''-shaped P-P plot.

\begin{figure*}[!ht]
	\centering
	\textbf{\hspace{1cm} BIAS \hspace{7.2cm} DISPERSION}\par\medskip
	\vspace*{-1mm}
	\includegraphics[width=0.2\textwidth]{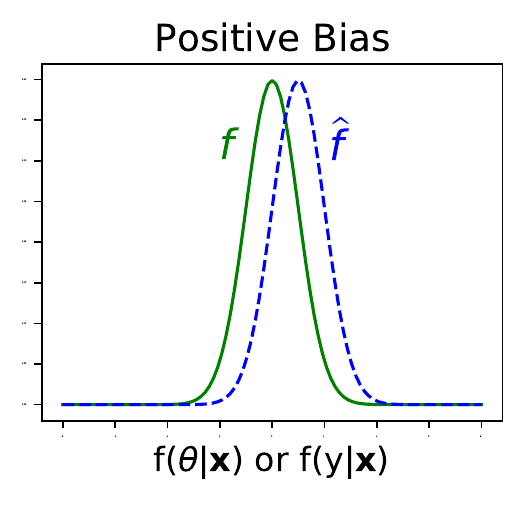}
	\hspace*{-4mm}
	\includegraphics[width=0.225\textwidth]{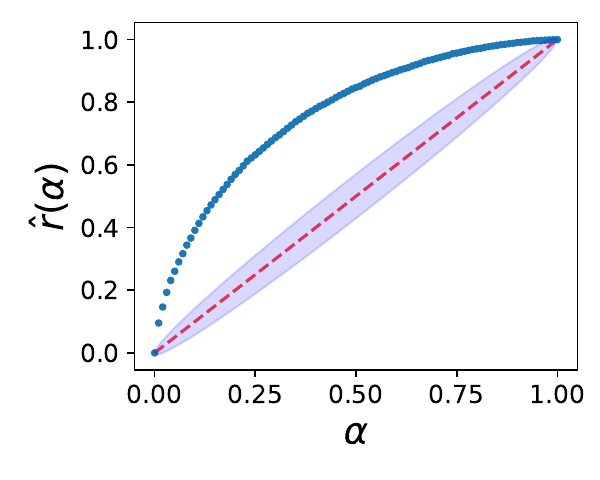}
	\hspace*{0.5cm}
	\includegraphics[width=0.2\textwidth]{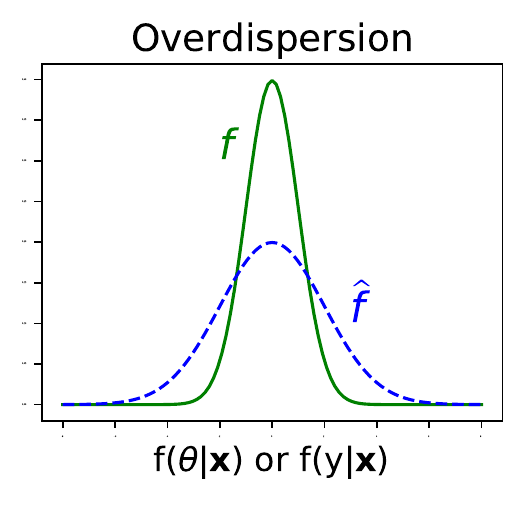}
	\hspace*{-4mm}
	\includegraphics[width=0.225\textwidth]{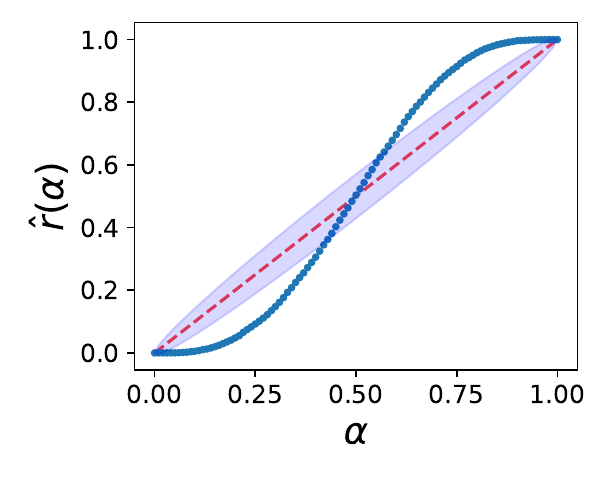}\\
	\includegraphics[width=0.2\textwidth]{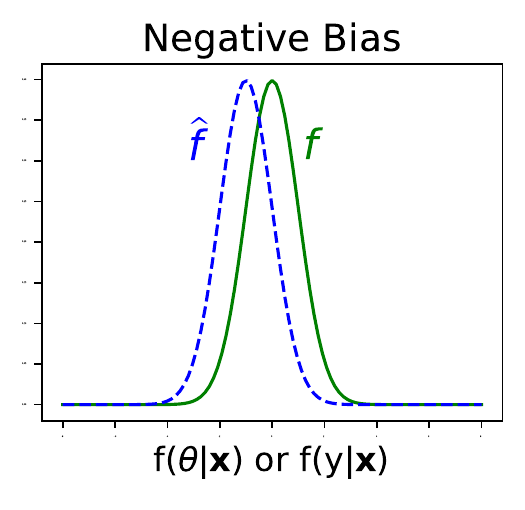}
	\hspace*{-4mm}
	\includegraphics[width=0.225\textwidth]{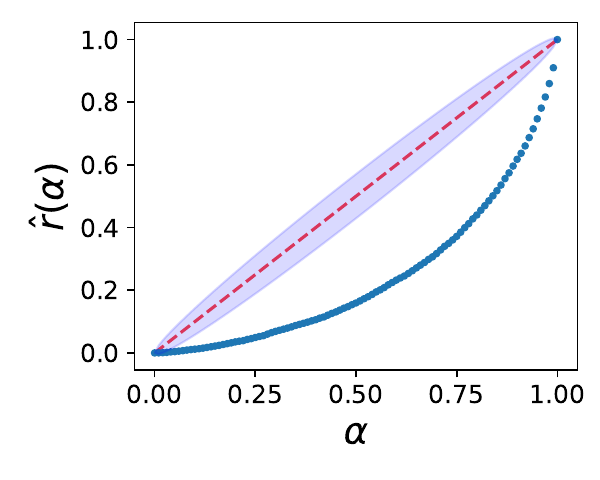}
	\hspace*{0.5cm}
	\includegraphics[width=0.2\textwidth]{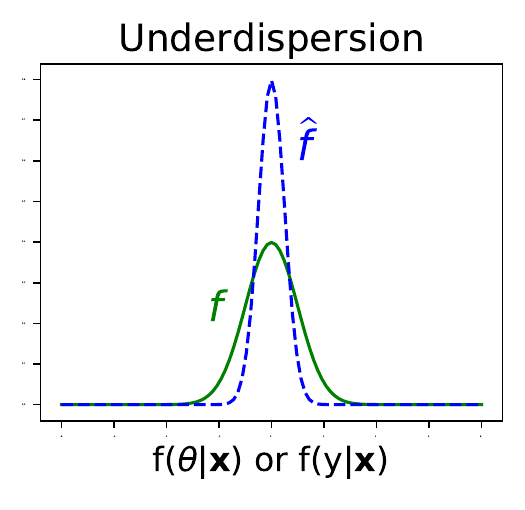}
	\hspace*{-4mm}
	\includegraphics[width=0.225\textwidth]{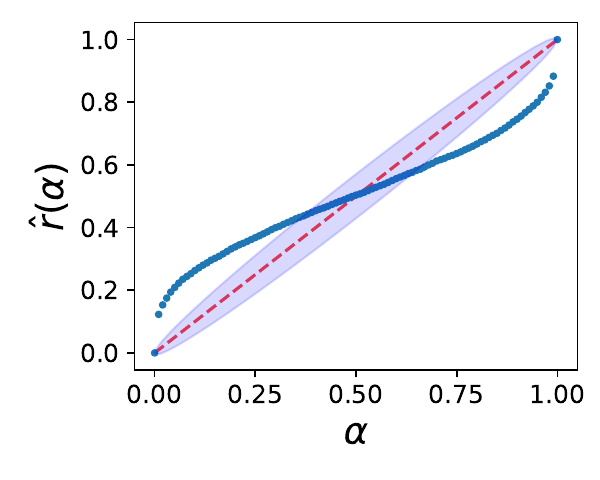}
	\caption{
		{\small P-P plots are commonly used to assess how well
			a density model fits actual data. Such plots display, in a clear and interpretable way, effects like bias (left panel) and dispersion (right panel) in an estimated distribution $\widehat f$ vis-a-vis 
			the true data-generating distribution $f$. Our framework yields a computationally efficient way to construct ``amortized local P-P plots''   for comparing conditional densities $\widehat f(\theta|\x)$ and $\widehat f(y|\x)$ at any location $\x$ of the feature space $\mathcal{X}$. See text for details and Sections \ref{sec:example1}-\ref{sec:example3} for examples.}
	}
	\label{fig:PPplot_interpret}
\end{figure*}

Of particular note is that our local P-P plots are  ``amortized'', in the sense that computationally expensive steps do not have to be repeated with e.g Monte Carlo sampling at each $\x$ of interest. Both the consistency tests in Section \ref{sec:our_tests} and the local P-P plots or ALPs only require initially training $\widehat r_{\alpha}$ on the observed data; the regression estimator can then be used to compute $\widehat r_{\alpha}(\x_{val})$ at any new evaluation point $\x_{val}$. Because of the flexibility in the choice of regression method, our construction also potentially scales to high-dimensional or different types of data $\x$. 
Algorithm \ref{alg:null_sample} details the construction of confidence bands for ALPs (under the null) using a Monte Carlo algorithm.
As an alternative to ALPs, one can also visualize the same information in local PIT histograms; 
see Algorithms 4-5
in Appendix~C.

\begin{algorithm}[!ht]
	\caption{
		Confidence band for ALP under $H_0$
	}\label{alg:null_sample}
	\algorithmicrequire \ {\small test data $\{\X_i\}_{i=1}^n$; test point $\x \in \mathcal{X}$; regression estimator $\widehat r$; grid $G$ of $\alpha$ values in $(0,1)$; number of null training samples $B$; confidence level 
		$1-\eta$}\\
	\algorithmicensure \ {\small  estimated 
		$(1-\eta)$ confidence band $\{L(\x), U(\x)\}$ for $\widehat r_{\alpha}(\x)$ under the null,
		for any $\x \in \mathcal{X}$, and $\alpha \in G$}
	\begin{algorithmic}[1]
		\State \codecomment{Recompute regression under null distribution:}
		\For{$b$ in $1,\ldots,B$}
		\State Draw $U_1^{(b)},\ldots,U_n^{(b)} \sim \textrm{Unif}[0,1]$.
		\For{$\alpha$ in $G$}
		\State Compute indicators $ \{ W^{(b)}_{\alpha,i} = \I( U_i^{(b)} < \alpha) \}_{i=1}^n $
		\State Train regression method $\widehat r_{\alpha}^{(b)}$ on $\{\X_i, W^{(b)}_{\alpha,i}\}_{i=1}^n$
		\EndFor
		\State Compute $\widehat r_{\alpha}^{(b)}(\x)$
		\EndFor
		\State \codecomment{Compute $(1-\eta)$ confidence band for $\widehat r_{\alpha}(\x)$:}
		\State $L(\x), U(\x) \gets \emptyset$
		\For{$\alpha$ in $G$}
		\State $L(\x) \gets L(\x) \cup \frac{\eta}{2}$-quantile of $\{\widehat r_{\alpha}^{(b)}(\x)\}_{b=1}^B$
		\State $U(\x) \gets U(\x) \cup (1-\frac{\eta}{2})$-quantile of $\{\widehat r_{\alpha}^{(b)}(\x)\}_{b=1}^B$
		\EndFor
		\\
		\textbf{return}  
		$L(\x),U(\x)$ on grid $G$
	\end{algorithmic}
\end{algorithm}

\subsection{Handling multivariate responses}
\label{sec:multivariate}

If the response $\Y$ is multivariate, then the random variable $F_{\Y|\X}(\Y|\X)$ is not uniformly distributed 
\citep{genest2001multivariate}, so \PIT\ values cannot be trivially generalized to higher dimensions.
One way to overcome this is to evaluate the \PIT\ statistic of univariate projections of $\Y$, as done by \citet{talts2018validating} for Bayesian consistency checks and \citet{mucesh2020machine} for the prediction setting. That is, the PIT values can be computed using the  estimate $\widehat f(h(\Y)|\x)$ induced by 
$\widehat f(\Y|\x)$ for some chosen $h: \mathbb{R}^p \longrightarrow \mathbb{R}$. Different projections can be used depending on the context. For instance, in Bayesian applications, posterior distributions are often used to compute credible regions for univariate projections of the parameters $\theta$. Thus, it is natural to evaluate PIT values of $h(\theta)=\theta_i$ for each parameter of interest. Another useful projection is copPIT \citep{ziegel2014copula}, which creates a unidimensional projection that has information about the joint distribution of $\Y$. Our diagnostic techniques are not enough to consistently assess the fit to $f(\Y|\x)$ if applied to these projections, but they do consistently evaluate the fit to $f(h(\Y)|\x)$, which is often good enough in practice. 

An alternative approach to assessing $\widehat f$ is through highest predictive density values
(\HPD \ values; \citealt{diana2015hpd,dalmasso2020conditional}), which are defined by
$$ \HPD(\y; \x) = \int_{{\y'}:\widehat{f}({\y'} \gvn \x) \ge \widehat{f}(\y \gvn \x)} \widehat{f}({\y'} \gvn \x) d{\y'}$$
(see Figure \ref{fig:pit_hpd_visualization}, bottom, for an illustration).
$\HPD(\y; \x)$ is a measure of how plausible $\y$ is according to $\widehat f(\y|\x)$ (in the Bayesian context, this is the complement of the e-value \citep{de1999evidence}; small values indicate high plausibility). As with \PIT\ values, \HPD\ values are uniform under the global null hypothesis 
\citep{dalmasso2020conditional}. However, standard goodness-of-fit tests based on  \HPD\ values share the same problem as those based on \PIT: they are 
insensitive to covariate transformations (see Theorem~\ref{thm:hpd_insensitive}, Supp. Mat.~A). Fortunately, 
\HPD \ values are
uniform under the local consistency hypothesis:
\begin{thm}
	\label{thm:hpd_conditional_unif}
	For any $\x \in \mathcal{X}$, if the local null hypothesis $H_0(\x):\widehat{f}(\cdot|\x)=f(\cdot|\x)$ holds, then the distribution of $\HPD(Y;\x)$ given $\x$ is uniform over $(0,1)$. (The reverse is however not true.)
\end{thm}
It follows that  the same techniques developed in Sections \ref{sec:our_tests} and \ref{sec:pp} can be used with \HPD \ values  to check global and local consistency for multivariate responses, as well as to construct local P-P plots. (Supp. Mat.~F showcases multivariate extensions via HPD.) The HPD statistic is especially appealing if one wishes to construct predictive regions with $\widehat f$ as
\HPD\ values are intrinsically related to highest predictive density sets 
\citep{hyndman1996computing}. HPD sets are region estimates of $\y$ that contain all $\y$'s for which $\widehat{f}(\y|\x)$ is larger than a certain threshold (in the Bayesian case, these are the highest posterior credible regions). More precisely, if $\HPD_{\alpha}(\x)$ is the $\alpha$-level HPD set for $\y$, then
$$\HPD(\y;\x)<\alpha \iff Y \in \HPD_{\alpha}(\x).$$
Thus, by testing local consistency of $\widehat f$ via \HPD\ values, we assess the coverage of \HPD \ sets. It should be noted, however, that even if the HPD values are uniform (conditional on $\x$), it may be the case that $\widehat f \neq f$.

\begin{figure}[htb]
	\centering
	\includegraphics[width=0.34\textwidth]{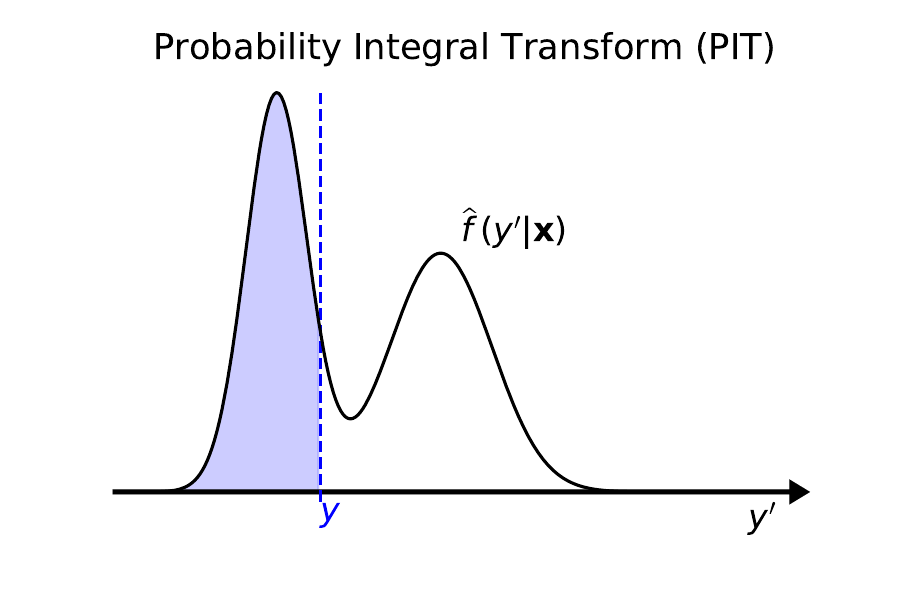}
	\includegraphics[width=0.34\textwidth]{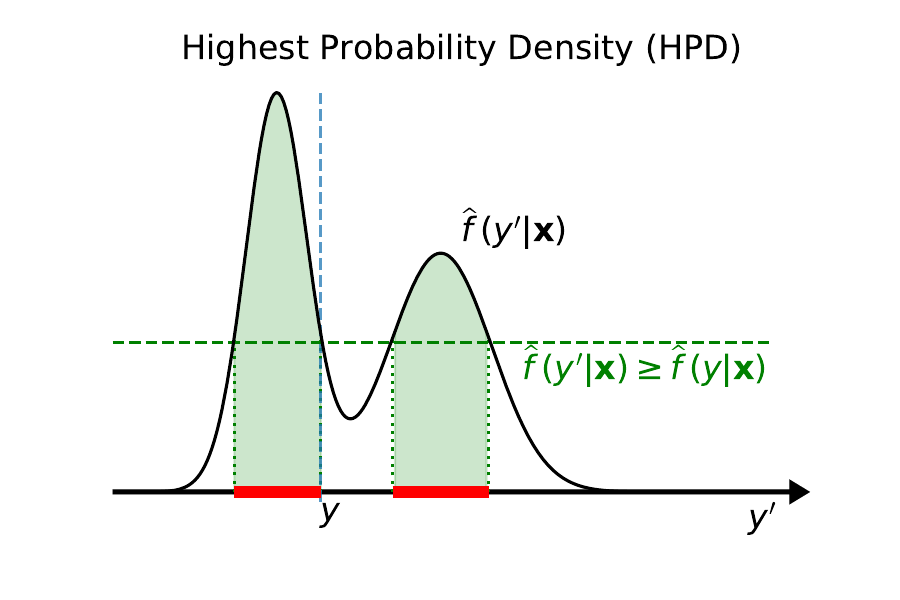}
	\caption{ 
		{\small Schematic diagram of the construction of PIT (top panel, shaded blue area) and HPD value (bottom panel, shaded green area) for an estimated density $\widehat{f}$ evaluated at $(y,\x)$. The highlighted red intervals in the bottom panel correspond to the highest density region (HDR) of $y|\x$.}
	}\label{fig:pit_hpd_visualization}
\end{figure}

\section{Example 1: Omitted Variable Bias in CDE Models}
\label{sec:example1}

Our first example involves omitted  but clearly relevant variables in a prediction setting. Inspired by Section 2.2.2 of \cite{cosma2013textbook}, we generate $\X = (X_1,X_2) \sim N(0, \Sigma) \in \mathbb{R}^2$, with $\Sigma_{1,1} = \Sigma_{2,2} = 1$ and $\Sigma_{1,2} = 0.8$, and take the response to be $Y|\X \sim N(X_1+X_2, 1)$. To mimic the variable selection procedure  common in high-dimensional inference methods, we fit two conditional density models: $\widehat f_1$, trained only on $X_1$, and $\widehat f_2$, trained on $\X$. Both models are fitted using a nearest-neighbor kernel CDE \citep{dalmasso2020conditional} with hyperparameters chosen by data splitting: we use 10000 training, 5000 validation, and 200 test points.

\begin{figure}[!ht]
	\centering
	\includegraphics[width=0.33\textwidth]{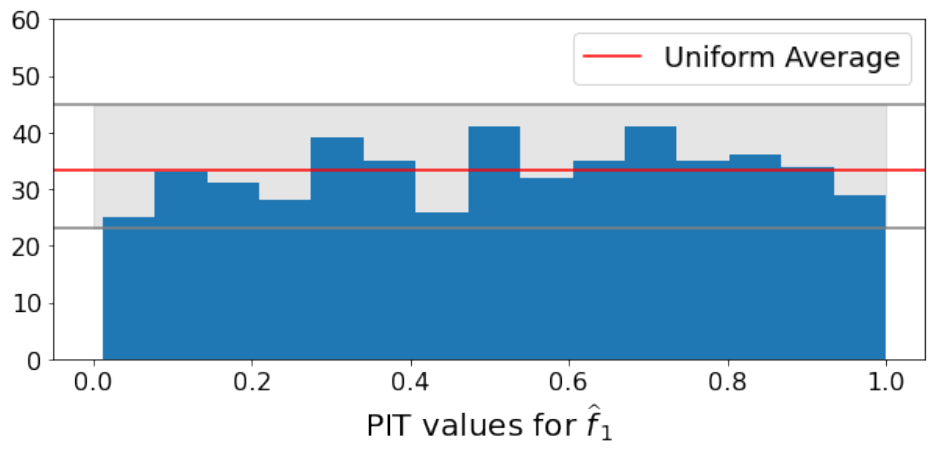}
	\vspace*{-3mm}
	\includegraphics[width=0.33\textwidth]{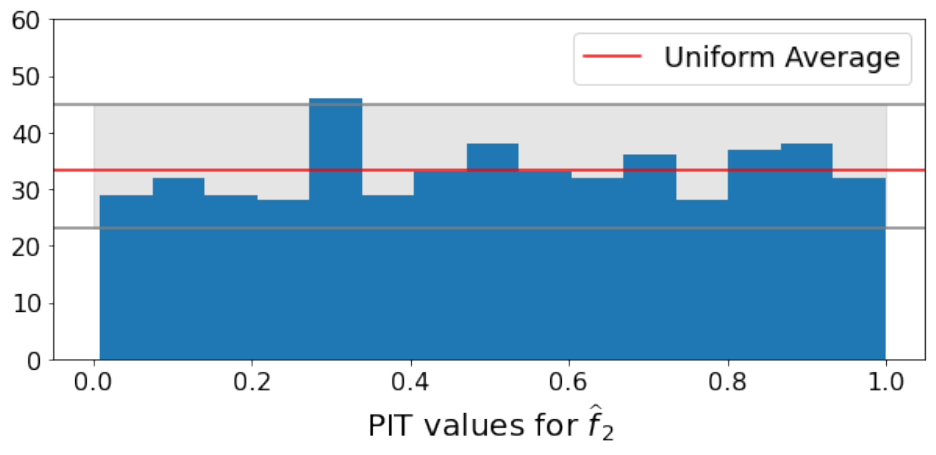}
	\caption{
		{\small Standard diagnostics for Example 1 showing histograms of PIT values computed on 200 test points (with 
			95\% confidence bands for a Unif[0,1] distribution). \textit{Top:} Results for $\widehat f_1$, which has only been fit to the first of two covariates.
			\textit{Bottom:} Results for $\widehat f_2$, which has been fit to both covariates.
			The top panel shows that standard PIT diagnostics  cannot tell that $\widehat f_1$ is a poor approximation to $f$. GCT, on the other hand,
			detects that $\widehat f_1$ is misspecified (p=0.004), while not rejecting the global null for $\widehat f_2$  (p=0.894).}
	}
	\label{fig:PIT_fail}
\end{figure}

\begin{figure*}[!ht]
	\centering
	\textbf{(a) \hspace{4cm} (b) \hspace{3cm} (c) \hspace{3cm} (d) }\par\medskip
	\includegraphics[width=1\textwidth]{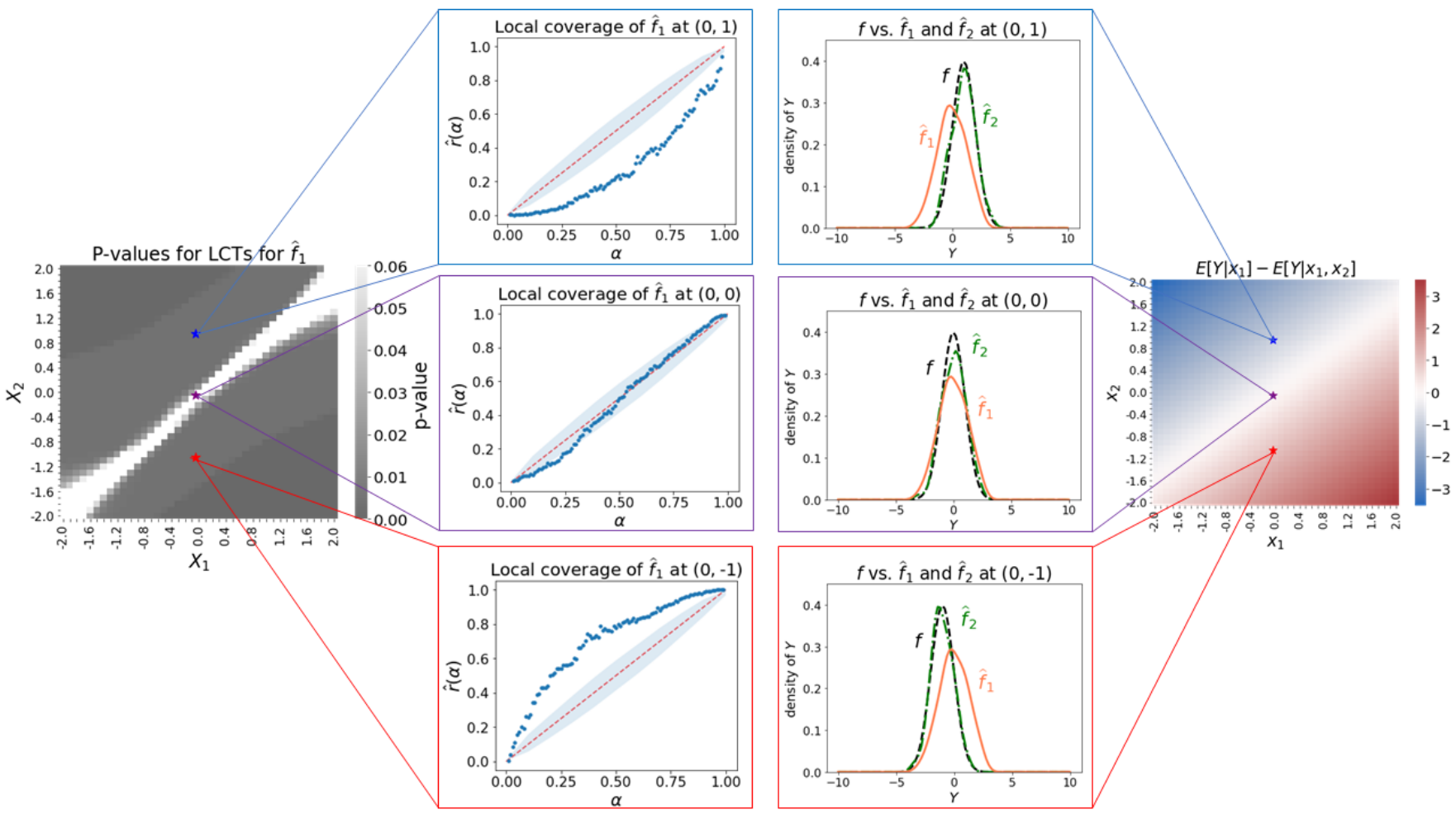}
	\caption{
		{\small New diagnostics for Example 1.
			\textbf{(a)} P-values for LCTs for $\widehat f_1$ indicate a poor fit across most of the feature space.
			\textbf{(b)} Amortized local P-P plots at selected points show the density $\widehat f_1$ as negatively biased (blue), well estimated 
			at significance level $\alpha=0.05$ with barely perceived overdispersion (purple), and positively biased (red). (Gray regions represent 95\% confidence bands under the null.)
			\textbf{(c)} $\widehat f_1$ and $\widehat f_2$ vs. the true (unknown) conditional density $f$ at the selected points. $\widehat f_1$ is clearly
			negatively and positively biased at the blue and red points, respectively, while the model does not reject the local null at the purple point. $\widehat f_2$ fits well at all three points.
			The difference on average in the predictions of $Y$ from $\widehat{f}_1(\cdot|\x)$ vs. the true distribution $f(\cdot|\x)$ for fixed $\x$ indeed corresponds to the ``omitted variable bias'' $\E[Y|x_1] - \E[Y|x_1,x_2]$.
			(\textit{Note:}
			Panels (c) and (d) require knowledge of the true $f$, which would not be available to the practitioner.)}
	}
	\label{fig:toy_example1_badfit}
\end{figure*}

This is a toy example where omitting one of the variables might lead to unwanted bias when predicting the outcome $Y$ for new inputs $\X$. As an indication of this bias, we have included a heat map (see panel (d) of Figure \ref{fig:toy_example1_badfit}) of the difference in the true (unknown) conditional means, $\E[Y|x_1] - \E[Y|x_1,x_2]$ as a function of $x_1$ and $x_2$. (In this example, the omitted variable bias is approximately the same as the difference in the averages of the predictions of $Y$ when using the model $\widehat f_1$ versus the model $\widehat f_2$ at any given $\x \in \mathcal{X}$; see Figure \ref{fig:toy_example1_badfit} panels (c) and (d)). Despite the clear relationship between $Y$ and $X_2$,  both $\widehat f_1$
(which omits $X_2$) and $\widehat f_2$ pass existing goodness-of-fit tests based on PIT (Figure \ref{fig:PIT_fail}). This result can be explained by Theorem \ref{thm:uniform}: because PIT is insensitive to covariate transformations and $\widehat f_1(y|\x) \approx f(y|x_1)$, \PIT \ values are  uniformly distributed,  even though $\widehat f_1$ omits a key variable. The GCT, however, detects that $\widehat f_1$ is misspecified ($p=0.004$), while the global null (Equation \ref{eq:global_null}) is not rejected for $\widehat f_2$ ($p=0.894$).

The next question a practitioner might ask is: ``What exactly is wrong with the fit?''. LCTs and local P-P plots can 
pinpoint the locations of discrepancies and describe the failure modes.
Panel (a) of Figure \ref{fig:toy_example1_badfit} shows p-values from local coverage tests for $\widehat f_1$ across the entire feature space of $\X$. The patterns in these p-values are largely explained by panel (d), which shows the difference between the conditional means of $Y$ given $x_1$ and given $x_1, x_2$. The detected level of discrepancy between the estimate $\widehat f_1$ and the true conditional density $f$ at a point $\x$ directly relates to the omitted variable bias $\E[Y|x_1] - \E[Y|x_1,x_2] = 0.8x_1 - x_2$:  the LCT p-values close to the line $x_2=0.8x_1$ are large (indicating no statistically significant deviations from the true model), and p-values decrease as we move away from this line.

Panel (b) of Figure \ref{fig:toy_example1_badfit} zooms in on a few different locations $\x$ with local P-P plots that depict and interpret distributional deviations. 
At the blue point, $\widehat f_1$ underestimates the true $Y$: we reject the local null (Equation \ref{eq:local_null}), 
and the P-P plot indicates negative bias. Conversely, at the red point, $\widehat f_1$ overestimates the true $Y$; we reject the local null, and the P-P plot indicates positive bias. At the purple point, $\widehat f_1$ is close to $f$, so the local null hypothesis is not rejected. 

This toy example is a simple illustration of the general  phenomenon of potentially unwanted omitted variable bias, which can be difficult to detect without testing for local and global consistency of models. Our proposed diagnostics identify this issue and provide insight into how the omitted variable distorts the fitted model relative to the true conditional density, across the entire feature space.

\begin{figure*}[!ht]
	\centering
	\includegraphics[width=0.92\textwidth]{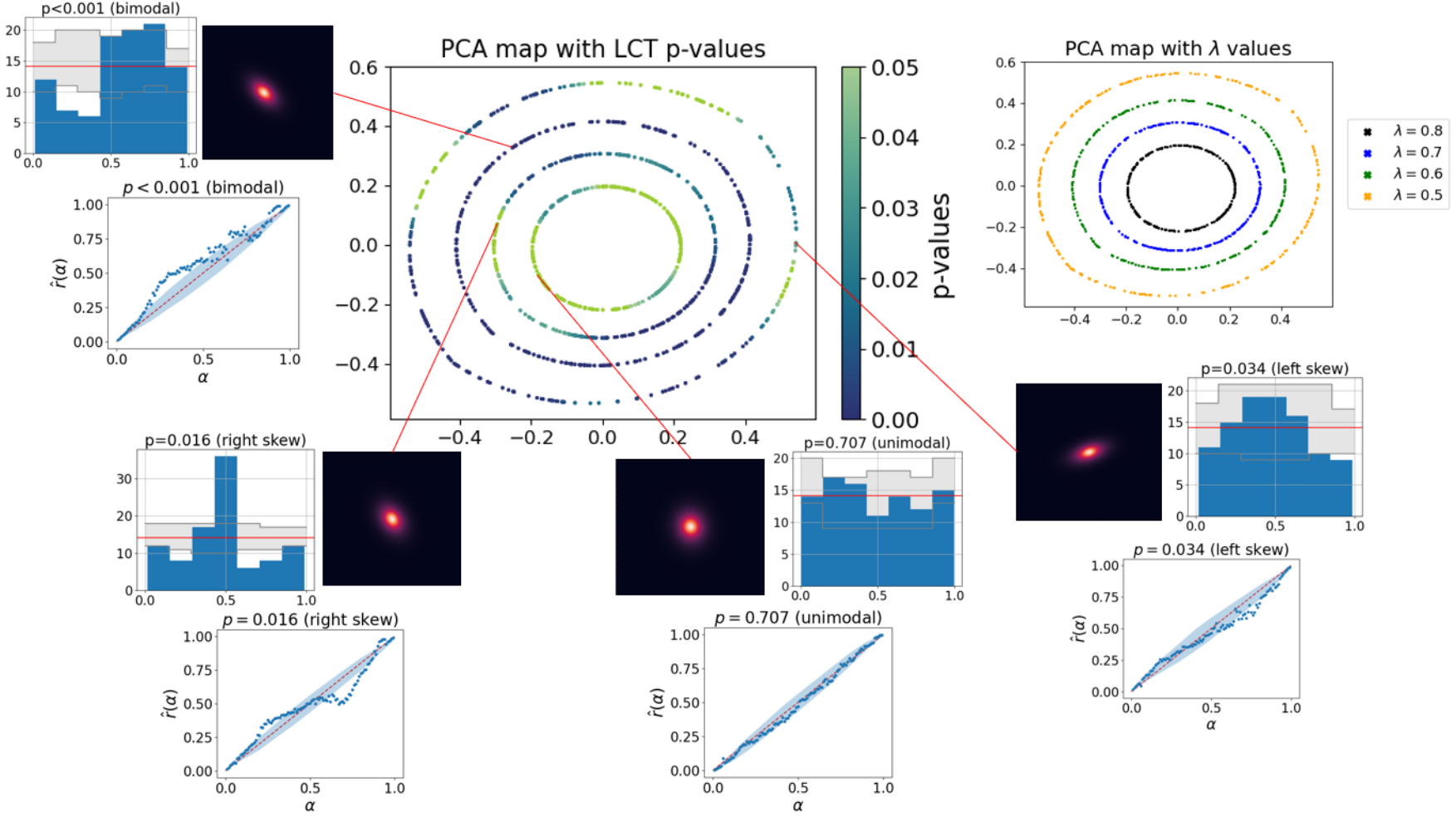}
	\caption{
		{\small New diagnostics for Example 2. For visualization, we show the location of the test galaxy points in $\mathbb{R}^{400}$ along the first two principal components (see center panel ``PCA map with LCT p-values''). Test statistics from the LCTs indicate that the unimodal density model generally fits well  for the $\lambda=0.8$ population, while fitting poorly for the other three populations with skewed and bimodal true redshift distributions. Local P-P plots or ALPs show statistically significant deviations in the CDEs (gray regions are 95\% confidence bands under the null) for the latter population, suggesting  the need for more flexible model classes. We also display local PIT histograms with confidence bands under the null, as a different way to present the same information as in the ALPs. (The histograms are computed from the $\widehat r_{\alpha}$ values according to Algorithm 4; no additional regression 
			is needed.)
		}
	}
	\label{fig:photoz_example}
\end{figure*}

\section{Example 2: Conditional Neural Densities for Galaxy Images}
\label{sec:example2}

In this example of CDE in a prediction setting, we apply neural density models to estimate the distribution of synthetic “redshift” $Z$ (a proxy for distance; the response) assigned to photometric or ``photo-z'' galaxy images $\X$ (the predictors). We then illustrate how our methods 
distinguish between ``good'' and ``bad'' CDEs. This toy example is motivated by the urgent need for metrics to assess photo-z probability density function accuracy. Diagnostics currently used by astronomers have known shortcomings \citep{schmidt2020photoz}, and our method is the first to properly address them.

Here, $\x$ represents a $20 \times 20$-pixel image of an elliptical galaxy generated by \texttt{GalSim}, an open-source toolkit for simulating realistic images of astronomical objects \citep{rowe2015galsim}. In \texttt{GalSim}, we can vary the axis ratio $\lambda$, defined as the ratio between the minor and major axes of the projection of the elliptical galaxy. We create four equally sized populations of galaxies, with $\lambda \in \{0.8, 0.7, 0.6, 0.5\}$. We then assign a response variable $Z$ according to different distributions (unimodal, skewed and bimodal) as follows:
\begin{align*}
Z | \lambda&=0.8 \sim N(0.1, 0.02) \\
Z | \lambda&=0.7 \sim \text{Beta}(3,7) \\
Z | \lambda&=0.6 \sim 0.6N(0.3, 0.05) + 0.4N(0.7, 0.05) \\
Z | \lambda&=0.5 \sim \text{Beta}(7,3).
\end{align*}
See Figure 8 in Supp. Mat. D for a plot of these distributions.

For illustration, 
we fit a unimodal Gaussian neural density model to  estimate the conditional density $Z|\X$. Our diagnostics pinpoint where 
in the feature space 
the density is bimodal or skewed, and thus a fit with one Gaussian is inadequate. We know of no other diagnostics that can provide such insight when fitting neural density models.
Specifically, we fit a convolutional mixture density network (ConvMDN, \cite{disanto2018cmdn}) with a single Gaussian component, two convolutional and two fully connected layers with ReLU activations \citep{glorot2011relu}. (We train on 10000 images using the Adam optimizer \citep{kingma2014adam} with learning rate $10^{-3}$, $\beta_1 = 0.9$, and $\beta_2 = 0.999$.) This gives an estimate of $f(z|\x)$.
We expect this CDE model to fit well for the $\lambda=0.8$ unimodal population, and fit poorly for the other bimodal or skewed  populations.

Our diagnostic framework effectively detects the flaws of this CDE model. First, we perform the GCT  which rejects  the global null ($p<0.001$). Next, we turn to LCTs and P-P plots to explore where and how the fit is inadequate. Figure \ref{fig:photoz_example} 
shows a principal component map of the test data.  The LCTs are able to identify a unimodal Gaussian model fits well for the $\lambda=0.8$ population, but that the same  model fails to  adequately estimate the  PDFs of the remaining populations.  P-P plots at selected test points indicate significant distributional deviations and suggest the need to consider more flexible model classes that incorporate bimodal and skewed distributions. 

\section{Example 3: Neural Posterior Inference for Galaxy Images}
\label{sec:example3}

\begin{figure*}[!ht]
	\centering
	\includegraphics[width=0.92\textwidth]{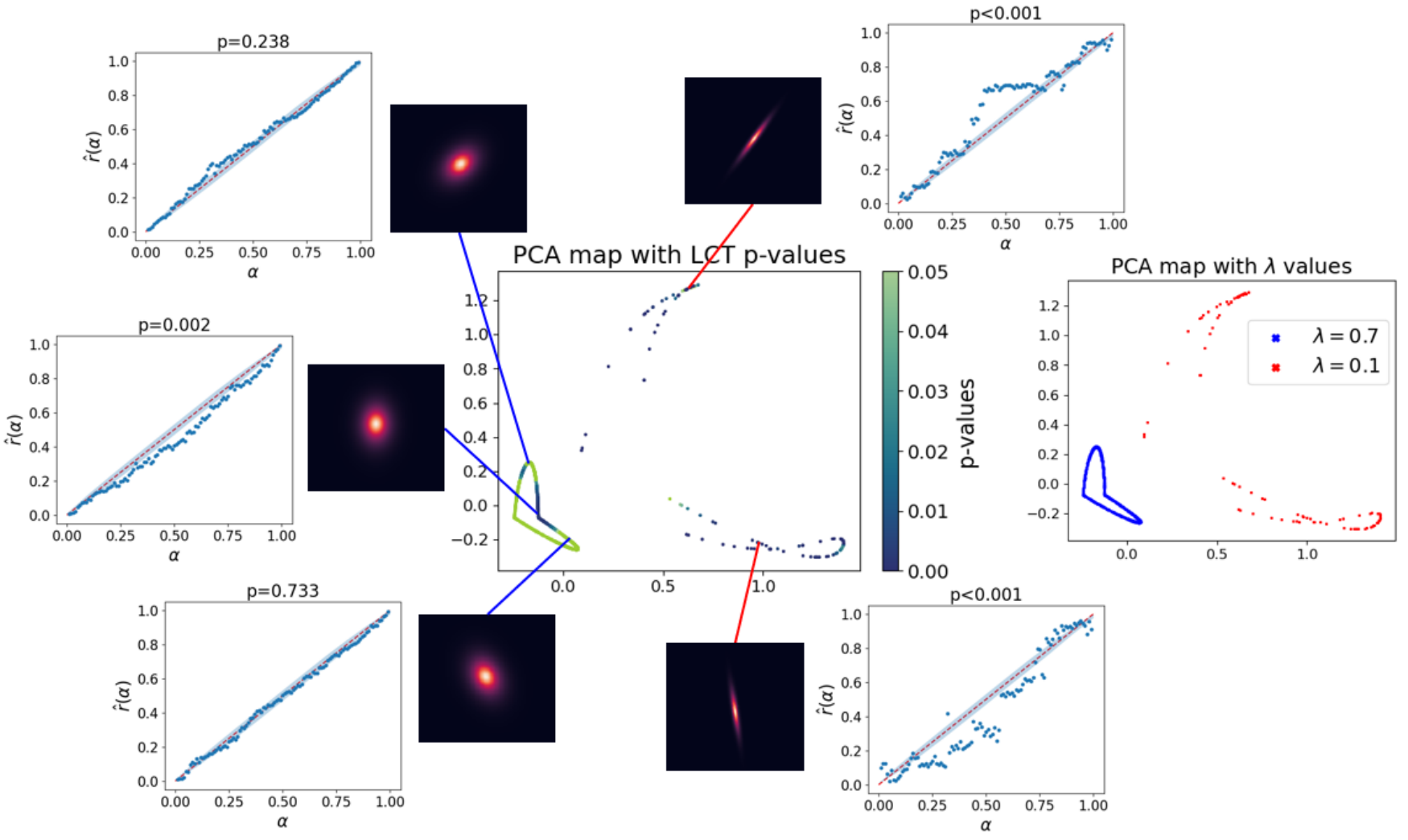}
	\caption{
		{\small New diagnostics for simulation-based inference algorithm in Example 3. For visualization, we show the location of the test galaxy points in $\mathbb{R}^{400}$ along its first two components  (see center panel ``PCA map with LCT p-values'').
			P-values for LCTs indicate that the ConvMDN  generally fits well for the dominant 90\% population of spheroidal galaxies
			($\lambda=0.7$), while fitting poorly for the smaller
			10\% subpopulation of elongated galaxies ($\lambda=0.1$).
			Local P-P plots
			show statistically significant deviations in the CDEs (gray regions are 95\% confidence bands under the null) for the latter population, suggesting we need better approximations of the posterior for this group.}
	}
	\label{fig:galsim_example}
\end{figure*}

Our final example tests for image data $\x \in \mathbb{R}^{400}$ whether a Bayesian posterior model $\widehat f(\theta|\x)$ fits the true posterior. As in Example 2, $\x$ represents an image of an elliptical galaxy generated by \texttt{GalSim}. As before, $\lambda$ is the galaxy's axis ratio, but now the quantity of interest $\theta$ is the galaxy's rotation angle with respect to the x-axis; that is, an unknown {\em internal} parameter. For illustration, we create a mixture of a larger population with $\lambda=0.7$ (spheroidal galaxies), and a smaller population with $\lambda=0.1$ (elongated galaxies). We then simulate a sample of images as follows: first, we draw $\lambda$ and $\theta$ from a prior distribution given by
\begin{align*}
\P(\lambda=0.7)&=1-\P(\lambda=0.1)=0.9 \\
&\theta\sim Unif(-\pi,\pi)
\end{align*}
Then we sample $20\times 20$ galaxy images $\X$ according to the data model $\X|\lambda,\theta \sim \texttt{GalSim}(a,\lambda)$, where 
\begin{align*}
a|\lambda&=0.7 \sim N(\theta,0.05) \\
a|\lambda&=0.1 \sim 0.5Laplace(\theta,0.05) + 0.5Laplace(\theta,0.0005).
\end{align*}

As in Example 2, we fit a convolutional mixture density network (ConvMDN); in this case, it gives us an estimate of the posterior distribution $f(\theta|\x)$. This time, we allow $K$, the number of mixture components, to vary. According to the KL divergence loss computed on a separate test sample with 1000 images, the best fit of $f(\theta|\x)$ is achieved by a ConvMDN model with $K=7$ (see Table 1 in Supp.~Mat.~E).
Here, the ConvMDN model with the smallest KL loss fails the GCT ($p<0.001$), so we turn to LCTs and P-P plots to understand why. Figure \ref{fig:galsim_example} plots the test galaxy images 
along their first two principal components. The LCTs show that the ConvMDN model
generally fits the density well for the main population of spheroidal galaxies ($\lambda=0.7$), but fails to properly model the smaller population of elongated galaxies ($\lambda=0.1$). P-P plots at selected test points indicate severe bias in the posterior estimates for the $\lambda=0.1$ population. These plots suggest that an effective way of obtaining a better approximation of the posterior is by improving the fit for the $\lambda=0.1$ population (by obtaining more data in that region of the feature space, using a different model class, etc). For instance, CDE models not based on mixtures \citep{papamakarios2019maf} could be more effective. 

{\bf Conclusion.} Conditional density models are widely used for uncertainty quantification in prediction and Bayesian inference. In this work, we offer practical procedures (GCT, LCT, ALP) for identifying, locating, and interpreting modes of failure for an approximation of the true conditional density. Our tools can be used in conjunction with loss functions, which are useful for performing model selection, but not good at evaluating whether a practitioner should keep looking for better models, or at providing information as to how a model could be improved. Finally, because LCT pinpoints hard-to-train regions of the feature space, our framework can provide guidance for active learning schemes.

{\bf Acknowledgments.} This work is supported by NSF DMS-2053804 and NSF PHY-2020295. RI is grateful for the financial support of CNPq (309607/2020-5) and FAPESP (2019/11321-9). 

\newpage

\bibliography{zhao_683}

\section*{Supplementary Materials}

\subsection*{A: Proofs}
\label{app:proofs}

In this section, we show proofs of the results stated in the paper.

\begin{proof}[Proof of Theorem \ref{thm:uniform}] 
	Let $z=g(\x)$ and $Z=g(\X)$. 
	Notice Equation \ref{eq:reduction} implies
	$\widehat F(Y|\x)= F(Y|g(\x))=F(Y|z)$, and thus
	\begin{align}
	\label{eq:equality_red}
	\widehat F(Y|\X)= F(Y|g(\X))=F(Y|Z)    
	\end{align}
	Thus, if $(\X,Y) \sim F_{\X,Y}$ then, for every $0 \leq a \leq 1$,
	\begin{align*}
	&\P(\PIT(Y,\X)\leq a) = \P(\widehat F(Y|\X)\leq a)\\
	&=\int_\mathcal{Z} \P(\widehat F(Y|\X)\leq a|Z=z)f(z)dz \\
	&=\int_\mathcal{Z} \P( F(Y|Z)\leq a|Z=z)f(z)dz \mbox{\ \ \ \ \ \ (Eq. \ref{eq:equality_red})}\\
	&=\int_\mathcal{Z} \P( F(Y|z)\leq a|Z=z)f(z)dz\\\
	&=\int_\mathcal{Z} \P( Y\leq F^{-1}(a|z) |Z=z)f(z)dz\\
	&=\int_\mathcal{Z}  F(F^{-1}(a|z)|Z=z)f(z)dz=\int_\mathcal{Z}  af(z)dz=a.
	\end{align*}
\end{proof}

\begin{proof}[Proof of Theorem \ref{thm:local_uniform}]
	Assume that $\widehat{f}(y|\x)=f(y|\x)$. It follows that,
	for any $0<\alpha<1$,
	\begin{align*}
	\P\left( \PIT(Y;\X) < \alpha|\x\right)&=\P \left(  F_{Y|\x}(Y) \leq \alpha |\x\right)\\
	&=\P \left(   Y \leq F ^{-1}_{Y|\x}(\alpha)|\x\right)\\
	&=F_{Y|\x}\left( F ^{-1}_{Y|\x}(\alpha)\right)\\
	&=\alpha,
	\end{align*}
	which shows that the distribution of $\PIT(Y;\X)$,  conditional on $\x$, is uniform.
	Now, assume that
	$\P\left(\PIT(Y;\X) < \alpha|\x\right)=\alpha$ for every $0<\alpha<1$ and let $\widehat F_{y|\x}(y)=\int_{-\infty}^{y}\widehat f(y'|\x)dy'$. Then
	\begin{align*}
	\alpha
	&= \P\left( \PIT(Y;\X) < \alpha|\x\right)\\
	&=\P \left( \widehat F_{Y|\x}(Y) \leq \alpha|\x\right)\\
	&=\P \left(   Y \leq \widehat F ^{-1}_{Y|\x}(\alpha)|\x\right)\\
	&=F_{Y|\x}\left( \widehat F ^{-1}_{Y|\x}(\alpha)\right).
	\end{align*}
	It follows that
	$F_{Y|\x}\left( \widehat F ^{-1}_{Y|\x}(\alpha)\right)=\alpha$, and thus
	$$\widehat F ^{-1}_{Y|\x}(\alpha)=F^{-1}_{Y|\x}\left(\alpha\right) \ \forall \alpha \in (0,1).$$
	The conclusion follows from the fact that the CDF characterizes the distribution of a random variable.
\end{proof}

\begin{proof}[Proof of Corollary \ref{cor:diagnostics}]
	Notice that
	$r_{\alpha}(\x)=\E\left[Z^\alpha|\x\right]= \P\left( \PIT(Y;\X) < \alpha|\x\right)$. It follows that $r_{\alpha}(\x)=\alpha$ 
	for every $\alpha \in (0,1)$ if, and only if, the distribution of $\PIT(\Y;\X)$, conditional on $\X$, is uniform over $(0,1)$. The conclusion follows from Theorem \ref{thm:local_uniform}.
\end{proof}

\setcounter{thm}{3}
\begin{thm}[\textbf{HPD values are insensitive to covariate transformations}]
	\label{thm:hpd_insensitive}
	Let $(\X, \Y) \sim F_{\X, \Y}$. If there exists a function $g:\mathcal{X} \to \mathcal{Z}$ such that $\widehat{f}(\y|\x) = f(\y|g(\x))$, then $\HPD(\Y;\X) \sim Unif(0,1)$.
\end{thm}

\begin{proof}[Proof of Theorem~\ref{thm:hpd_insensitive}]
	Under the assumption we can rewrite the HPD value as:
	
	\begin{align*}
	\HPD(\y,\x) &= \int_{\y': f(\y'|g(\x)) > f(\y|g(\x))} f(\y'|g(\x)) dy' \\
	&= \int_{y': f(\y'|\z) > f(\y|\z)} f(\y'|\z) dy' = \HPD(\y,\z),
	\end{align*}
	
	with $g(\x) = \z$. Following the proof structure by \cite{diana2015hpd} closely, we define the random variable $\xi_{\z,\y}=\HPD(\z, \y)$, equipped with the probability density function $h: (\mathcal{Z} \times \mathcal{Y}) \to \mathbb{R}$. Dropping the subscripts for simplicity, let $\xi^*=\HPD(\z^*, \y^*)$ the HPD value of a specific pair $(\z^*, \y^*)$; $\xi^*$ is the probability mass of $f$ above the level set $f(\y^*|\z^* = g(\x^*))$. Without loss of generality, if we show that $h(\xi^*) = 1$ we can conclude that $\xi(y,z)$ is uniformly distributed $U[0,1]$. Using the fundamental theorem of calculus we can write:
	
	\setcounter{algorithm}{2}
	\begin{algorithm}[!ht]
		\caption{
			P-values for Global Coverage Test
		}\label{alg:p_values_gct}
		\algorithmicrequire \ {\small  conditional density model $\widehat f$; test data $\{\X_i,Y_i\}_{i=1}^n$; regression estimator $\widehat r$; number of null training samples $B$}\\
		\algorithmicensure \ {\small estimated p-value $\widehat{p}(\x)$ across all $\x \in \mathcal{X}$}
		\begin{algorithmic}[1]
			\State \codecomment{Compute test statistic over $\X_1,\ldots,\X_n$:}
			\State Compute values $\PIT(Y_1;\X_1),\ldots,\PIT(Y_n;\X_n)$
			\State $G \leftarrow$ grid of $\alpha$ values in $(0,1)$.
			\For{$\alpha$ in $G$}
			\State Compute indicators $Z^{\alpha}_1,\ldots,Z^{\alpha}_n$
			\State Train regression method $\widehat r_{\alpha}$ on $\{\X_i, Z^{\alpha}_i\}_{i=1}^n$
			\EndFor
			\State Compute test statistic $S = \frac{1}{n} \sum_{i=1}^n T(\X_i)$
			\State \codecomment{Recompute test statistic under null distribution:}
			\For{$b$ in $1,\ldots,B$}
			\State Draw $U_1^{(b)},\ldots,U_n^{(b)} \sim \textrm{Unif}[0,1]$.
			\For{$\alpha$ in $G$}
			\State Compute indicators $ \{ Z^{(b)}_{\alpha,i} = \I( U_i^{(b)} < \alpha)\}_{i=1}^n $
			\State Train regression method $\widehat r_{\alpha}^{(b)}$ on $\{\X_i, Z^{(b)}_{\alpha,i}\}_{i=1}^n$
			\EndFor
			\State Compute $\displaystyle T^{(b)}(\X_i):=\frac{1}{|G|}\sum_{\alpha \in G} (\widehat r_{\alpha}^{(b)}(\X_i)-\alpha)^2$ for $i=1,\ldots,n$
			\State Compute $S^{(b)} := \frac{1}{n} \sum_{i=1}^n T^{(b)}(\X_i)$
			\EndFor
			\\
			\textbf{return}  $\displaystyle \widehat{p}(\x):=\frac{1}{B}\sum_{b=1}^B\I\left(S < S^{(b)}\right)$
		\end{algorithmic}
	\end{algorithm}
	
	\begin{align*}
	h(\xi^*) &= \frac{\partial}{\partial \xi^*} \int_{-\infty}^{\xi^*} g(\epsilon) d\epsilon \\ 
	&= \frac{\partial}{\partial \xi^*} \int_{-\infty}^{\xi^*} \int_{\mathcal{Z} \times \mathcal{Y}} \delta(\xi(y, z) -  \epsilon) dF(z,y) d\epsilon \\
	&= \frac{\partial}{\partial \xi^*} \int_{\mathcal{Z} \times \mathcal{Y}} \Phi(\xi(y, z) - \xi^*) dF(z,y) \\
	&= \frac{\partial}{\partial \xi^*} \int_{\mathcal{Z}} \left[ \int_{\mathcal{Y}} \Phi(\xi(y, z) - \xi^*) f(y|z) dy \right] f(z) dz \\
	&= \frac{\partial}{\partial \xi^*} \int_{\mathcal{Z}} \xi^* f(z) dz = \frac{\partial}{\partial \xi^*} \xi^{*} = 1
	\end{align*}
	
	where $\Phi$ is the Heavyside function, which is $1$ when the argument is positive and $0$ otherwise.
\end{proof}

\begin{proof}[Proof of Theorem~\ref{thm:hpd_conditional_unif}]
	Under the null hypothesis $H_0(\x)$ for any $\x \in \mathcal{X}$ we have that:
	
	\begin{align}
	\HPD(\y;\x) &= \int_{{\y'}:\widehat{f}({\y'} \gvn \x) \ge \widehat{f}(\y \gvn \x)} \widehat{f}({\y'} \gvn \x) d{\y} \\
	&= \int_{{\y'}:f({\y'} \gvn \x) \ge f(\y \gvn \x)} f({\y'} \gvn \x) d{\y}.
	\end{align}
	
	Applying the results about uniformity of HPD for $f(\cdot|\x)$ from \citet[Section A.2]{diana2015hpd} (also reproduced in the proof of Theorem~\ref{thm:hpd_insensitive}) proves the theorem.
	
\end{proof}

\begin{table*}[ht!]
	\centering
	\resizebox{0.65\textwidth}{!}{%
		\begin{tabular}{|c|c|c|c|c|c|c|c|c|c|}
			\hline
			\textbf{K} & \textbf{2} & \textbf{3} & \textbf{4} & \textbf{5} & \textbf{6} & \textbf{7} & \textbf{8} & \textbf{9} & \textbf{10} \\
			\hline
			\textbf{KL loss} & -0.729 & -0.885 & -0.915 & -0.906 & -0.897 & -0.917 & -0.906 & -0.911 & -0.905 \\
			\hline
	\end{tabular}}
	\caption{{\small
			The KL divergence loss indicates that the number of mixture components in the ConvMDN approximation of the posterior in Example 2 should be $K=7$. 
	}}
	\label{tab:KL_loss}
\end{table*}

\subsection*{B: GLOBAL COVERAGE TEST}
\label{app:gct}

Algorithm \ref{alg:p_values_gct} describes our procedure for testing global consistency (see Definition 1 in the paper) using a Monte Carlo sampling strategy.

\subsection*{C: LOCAL PIT HISTOGRAMS}
\label{app:histograms}

\begin{algorithm}[!ht]
	\caption{
		Local PIT histograms 
	}\label{alg:histograms}
	\algorithmicrequire \ {\small conditional density model $\widehat f$; test data $\{\X_i,Y_i\}_{i=1}^n$; test point $\x \in \mathcal{X}$; regression estimator $\widehat r$; grid $G$ of $\alpha$ values in $(0,1)$; 
		number of bins $n_{\text{bin}}$}\\
	\algorithmicensure \ {\small local PIT histogram $H(\x)$ for any $\x \in \mathcal{X}$}
	\begin{algorithmic}[1]
		\State \codecomment{Compute estimated local coverage at $\x$:}
		\State Compute values $\PIT(Y_1;\X_1),\ldots,\PIT(Y_n;\X_n)$
		\For{$\alpha$ in $G$}
		\State Compute indicators $W^{\alpha}_1,\ldots,W^{\alpha}_n$
		\State Train regression method $\widehat r_{\alpha}$ on $\{\X_i, W^{\alpha}_i\}_{i=1}^n$
		\EndFor
		\State Compute values $\left\{ \widehat r_{\alpha}(\x) \right\}_{\alpha \in G}$
		\State \codecomment{Compute local PIT histogram:}
		\State Create histogram $H(\x)$ of $\left\{ \widehat r_{\alpha}(\x) \right\}_{\alpha \in G}$ values by dividing $[0,1]$ into  $n_{\text{bin}}$ equal-sized bins
		\\
		\textbf{return} histogram $H(\x)$
	\end{algorithmic}
\end{algorithm}

Algorithm \ref{alg:histograms} describes our procedure for constructing local PIT histograms. Note that if one has already obtained estimators $\widehat r_{\alpha}$ of 
the local PIT distribution (cdf) via regression (by, for example, running Algorithm \ref{alg:p_values}), then one can generate a local histogram 
at any $\x \in \mathcal{X}$ by simply 
using those $\widehat r_{\alpha}$ functions at $\x$, without needing to rerun any regressions. 
Similarly, there is no need to repeat the MC sampling under the null in Algorithm 2 to create confidence bands for the local PIT histograms.

\begin{algorithm}[!ht]
	\caption{Confidence band for local PIT histogram  under $H_0$ 
	}\label{alg:null_sample_PIT}
	\algorithmicrequire \ {\small test data $\{\X_i\}_{i=1}^n$; test point $\x \in \mathcal{X}$; regression estimator $\widehat r$; number of bins $n_{\text{bin}}$; grid $G$ of $\alpha$ values in $(0,1)$; number of null training samples $B$; confidence level 1-$\eta$; number of bins $n_{\text{bin}}$}\\
	\algorithmicensure \ {\small estimated $(1-\eta)$ confidence band $\{L(\x), U(\x)\}$ for local PIT histogram $H(\x)$  under the null, for any $\x \in \mathcal{X}$}
	\begin{algorithmic}[1]
		\State \codecomment{Recompute regression under null distribution:}
		\For{$b$ in $1,\ldots,B$}
		\State Draw $U_1^{(b)},\ldots,U_n^{(b)} \sim \textrm{Unif}[0,1]$.
		\For{$\alpha$ in $G$}
		\State Compute indicators $ \{ W^{(b)}_{\alpha,i} = \I( U_i^{(b)} < \alpha) \}_{i=1}^n $
		\State Train regression method $\widehat r_{\alpha}^{(b)}$ on $\{\X_i, W^{(b)}_{\alpha,i}\}_{i=1}^n$
		\EndFor
		\EndFor
		\State \codecomment{Compute 
			confidence band:} 
		\For{$b$ in $1,\ldots,B$}
		\State Create histogram $H^{(b)}(\x)$ of $\left\{\widehat r_{\alpha}^{(b)}(\x) \right\}_{\alpha \in G}$ values 
		by dividing $[0,1]$ into  $n_{\text{bin}}$ equal-sized bins 
		\EndFor
		
		\State  $L(\x) \leftarrow \frac{\eta}{2}$-quantile of $\{H^{(b)}(\x) \}_{b=1}^B$
		\State   $U(\x) \leftarrow (1-\frac{\eta}{2})$-quantile of $\{  H^{(b)}(\x) \}_{b=1}^B$
		\State \textbf{return} $L(\x), U(\x)$
	\end{algorithmic}
\end{algorithm}

\subsection*{D: EXAMPLE 1: OMITTED VARIABLE BIAS IN CDE MODELS}
\label{app:figures}

In this section we show the results of the local test on Example 1 for model $\widehat f_2$, which passes the global test.

Figure~\ref{fig:toy_example1_goodfit}, right panel, shows p-values from LCTs across the feature space for the model $\widehat{f}_2$. Unlike model $\widehat{f}_1$, which was fit on $X_1$ alone, $\widehat{f}_2$ was fit on both $X_1$ and $X_2$. Hence, $\widehat{f}_2$ is able to pass all tests, with local P-P plots indicating a good fit (with two examples shown in the Figure~\ref{fig:toy_example1_goodfit}, left panel).

\setcounter{figure}{6}
\begin{figure}[!ht]
	\centering
	\includegraphics[width=0.45\textwidth]{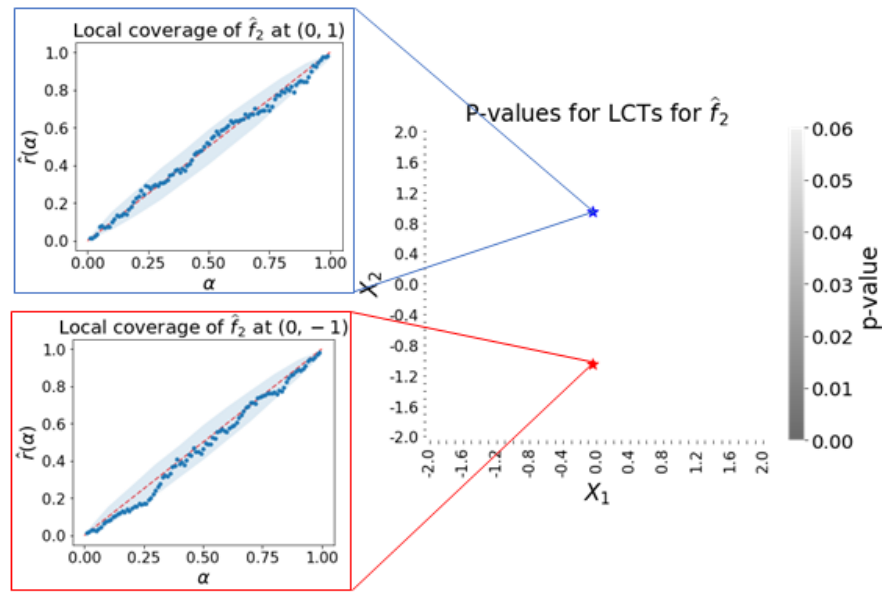}
	\caption{
		{\small P-values for LCTs for $\widehat f_2$ in Example 1 suggest an adequate fit everywhere in the feature space; local coverage plots at selected points also suggest a good fit.}
	}
	\label{fig:toy_example1_goodfit}
\end{figure}

\subsection*{E: Example 2: Conditional Neural Density Modeling for Galaxy Images}

Figure \ref{fig:z_dists} shows the true conditional densities of the simulated ``redshift'' $Z$ vs. the axis ratio $\lambda$ of the corresponding galaxy image.

\begin{figure}[!ht]
	\centering
	\includegraphics[width=0.45\textwidth]{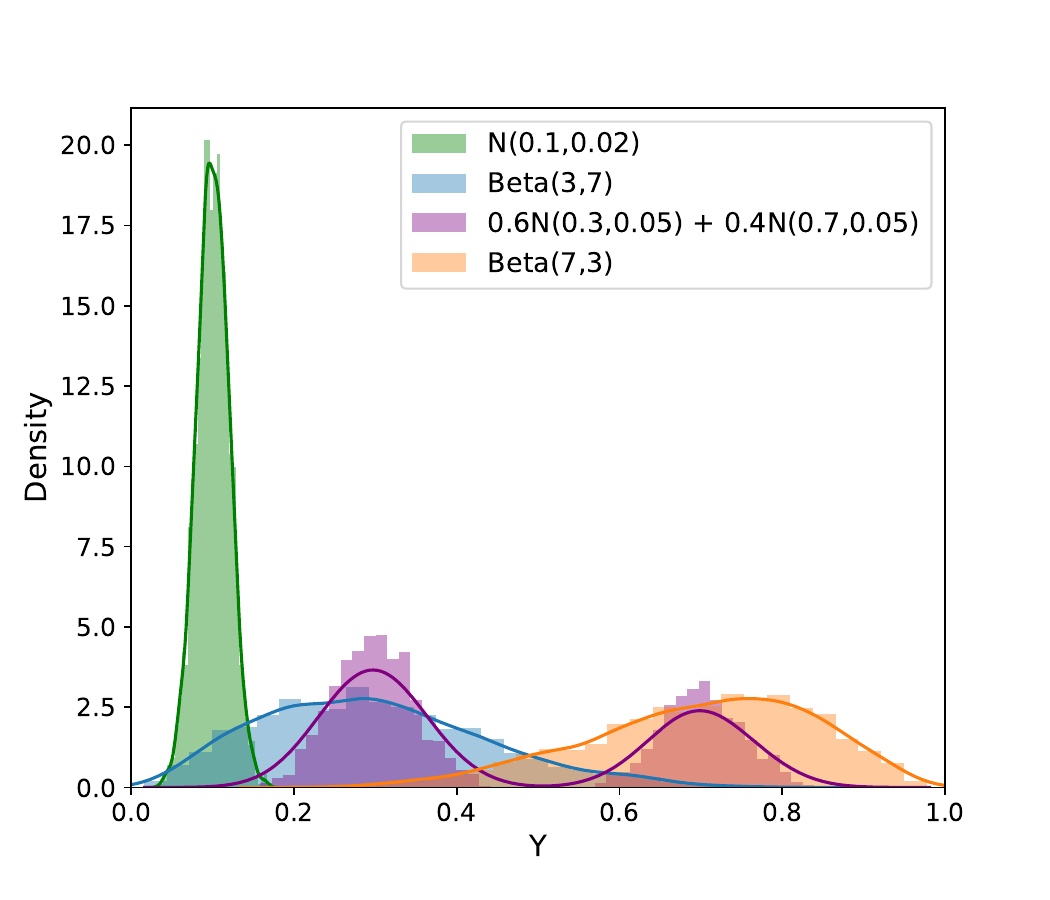}
	\caption{{\small We assign a unimodal distribution of ``redshift'' $Z$ for to the galaxy population with $\lambda=0.8$, and higher, more skewed and bimodal distributions of $Z$ to the populations with $\lambda=0.7,0.6,0.5$.}
	}
	\label{fig:z_dists}
\end{figure}

\subsection*{F: Example 3: Posterior Inference for Galaxy Images}

Table~\ref{tab:KL_loss} reports the KL divergence loss over a test set of 1000 galaxy images for a ConvMDN model with $K$ components, for $K=2,...,10$. The KL loss indicates that  $K=7$ is the optimal choice. However, in the paper we show that this model fails to pass our GCT and therefore is not a good approximation of the true conditional density. Figure~\ref{fig:galsim_example} in the paper also shows how to use our LCTs and P-P plots to diagnose the inadequacies in the fit.

\subsection*{G: Example 4: Conditional density models with multivariate response}
\label{app:multivariate}
For multivariate response $\Y$, we can assess the quality of fit of $\widehat f$ through highest predictive density (HPD) values, as described in Section \ref{sec:multivariate}. Our method still yields interpretable diagnostics, but the interpretation of HPD values differs from that of PIT values. If a local P-P plot shows estimated HPD values $\widehat r_{\alpha}$ that are too high relative to $\alpha$, this suggests that the model is overdispersed relative to the true density. HPD values that are too low could suggest an underdispersed model, or be a symptom of model misspecification: if the estimated density is systematically biased (i.e. not centered at the same location as the true density), the observed values $Y$ will disproportionately represent lower density contours of the true density.

In this example, we draw $\X = (X_1,X_2) \sim \textrm{Unif}[0,1]^2$, and then define a bivariate response $\Y = (Y_1,Y_2)$ as follows:
\begin{align*}
\Y|\X \sim
\begin{cases}
N((X_1,X_2), I_2), & X_2 \in [1,2] \\
N((X_1,X_2), 0.25I_2), & X_2 \in [0,1] \\
t_4 \textrm{ centered at } (X_1,X_2), & X_2 \in [-1,0] \\
t_4 \textrm{ centered at } (X_1+1,X_2+1), & X_2 \in [-2,-1] \\
\end{cases}
\end{align*}
where $I_2$ is the identity matrix. See Figure \ref{fig:true_densities_multivar} for an illustration of how the true conditional density $f(\y|\x)$ varies across the feature space. For illustration, we choose the model $\widehat f(\cdot | \x) = N((x_1,x_2), 1)$ in all four regions. This model perfectly fits the true density when $x_2 \in [1,2]$, and is misspecified in the other cases. We evaluate HPD values at 1000 test points to run our diagnostic framework.

Figure \ref{fig:example4_multivariate} summarizes the results of our diagnostics. First, we perform the GCT, which rejects the global null with $p<0.001$. We then perform LCTs across the feature space for $\X$; the resulting p-values are shown in the center panel. As expected, LCTs indicate a good fit when $\widehat f$ is correct, and a poor fit in most regions where $\widehat f$ is misspecified. Investigating further with local P-P plots enables us to detect overcoverage and undercoverage of HPD regions at specific locations in the feature space. Overcoverage of the true $\Y$ by the HPD region means the $\alpha$-HPD set for $\widehat f$ is too large, so observed HPD values are too low: this indicates that $\widehat f$ is overdispersed locally (as in the top right example). Conversely, undercoverage by the HPD region means the $\alpha$-HPD set for $\widehat f$ does not cover enough of the true density mass of $f$, so observed HPD values are too high: this can be caused by $\widehat f$ being underdispersed or biased locally (as in the bottom right example).

\clearpage

\begin{figure*}[!ht]
	\centering
	\includegraphics[width=0.4\textwidth]{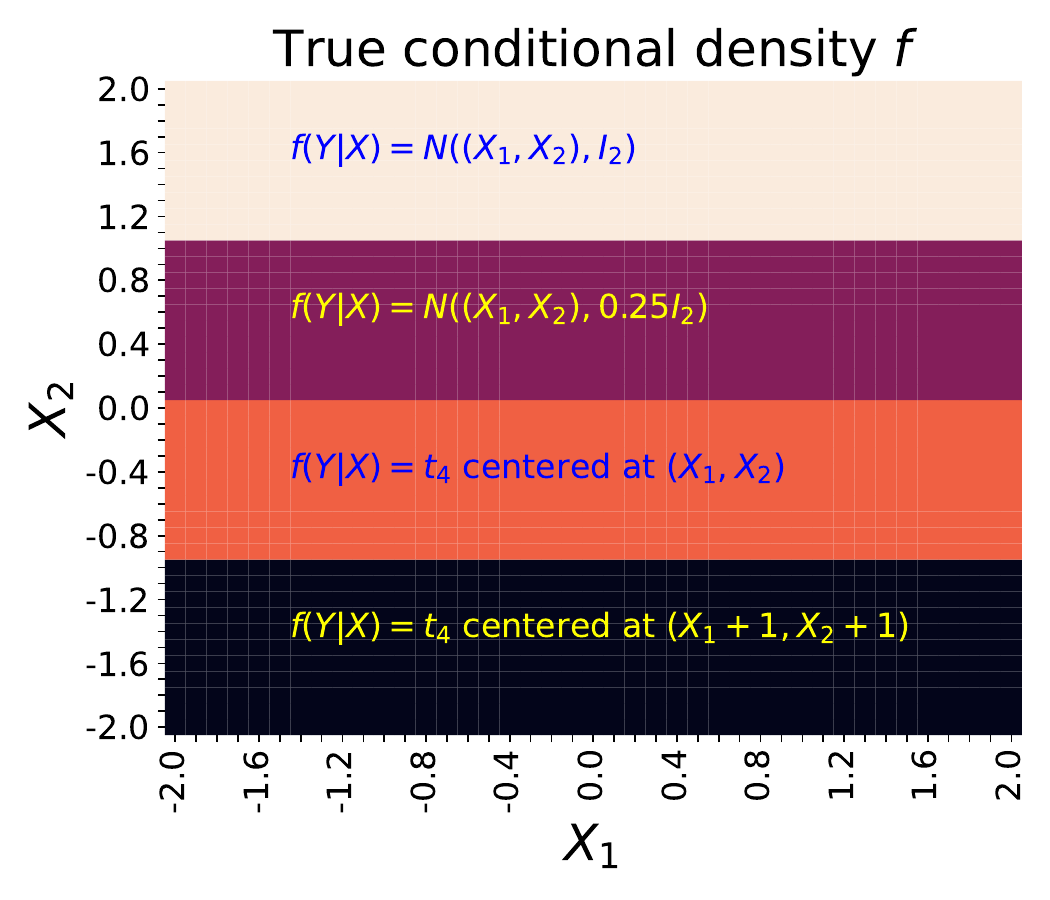}
	\caption{
		\small The true conditional density $f(\y|\x)$ has different forms in four different regions of the feature space, whereas we assume the same model $\widehat{f}(\y|\x) = N((x_1,x_2), 1)$ across feature space. When $X_2 \in [1,2]$, the model $\widehat f$ is correctly specified. When $X_2 \in [0,1]$, $\widehat f$ is overdispersed relative to the true density $f$. When $X_2 \in [-1,0]$, $\widehat f$ is slightly underdispersed relative to the true density $f$. When $X_2 \in [-2,-1]$, $\widehat f$ is both biased and slightly underdispersed relative to the true density $f$.
	}
	\label{fig:true_densities_multivar}
\end{figure*}

\begin{figure*}[!ht]
	\centering
	\includegraphics[width=1\textwidth]{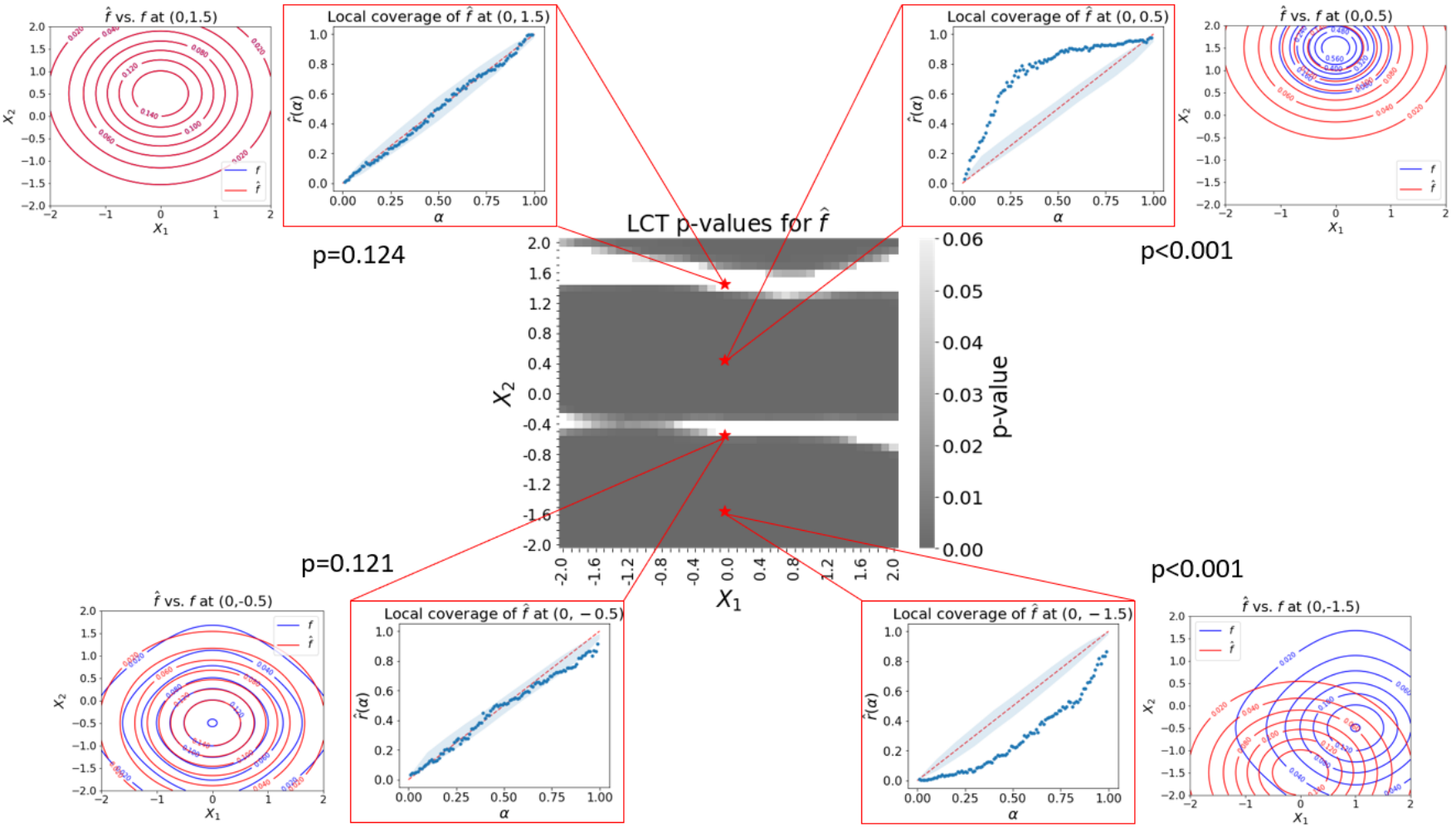}
	\caption{
		{\small New diagnostics for Example 4.
			P-values for LCTs for $\widehat f$ indicate a poor fit for values of $X$ where $X_2 \in [0,1]$ or $X_2 \in [-2,-1]$ (see center panel).
			Amortized local P-P plots at selected points show the HPD level sets of $\widehat f$ as overdispersed for $X_2 \in [0,1]$, and underdispersed or biased for $X_2 \in [-2,-1]$. In contrast, the HPD level sets are well estimated at significance level $\alpha=0.05$ for $X_2 \in [1,2]$ and $X_2 \in [-1,0]$. (Gray regions represent 95\% confidence bands under the null.)
			Contour plots show the model $\widehat f$ vs. the true (unknown) conditional density $f$ at the selected points. $\widehat f$ is clearly overdispersed at $(0,0.5)$ and systematically biased at $(0,-1.5)$. The model perfectly fits the density at $(0,1.5)$, and has barely detectable underdispersion at $(0,-0.5)$.
			(\textit{Note:}
			The contour plots requires knowledge of the true $f$, which would not be available to the practitioner.)}
	}
	\label{fig:example4_multivariate}
\end{figure*}

\end{document}